\documentclass[10pt,journal,twoside]{IEEEtran}

\usepackage[numbers,sort&compress]{natbib}
\usepackage{graphicx}
\usepackage[caption=false]{subfig}
\usepackage{verbatim}
\usepackage{multirow}
\usepackage{amsmath}
\usepackage{amsfonts}
\usepackage{amssymb}
\usepackage{dsfont}
\usepackage{setspace}
\usepackage[ruled,vlined,linesnumbered]{algorithm2e}

\usepackage{balance}
\usepackage{amsthm}
\newtheorem{theorem}{Theorem}
\newtheorem{remark}{\indent Remark}
\newtheorem{lemma}[theorem]{Lemma}

\newtheorem{assumption}{Assumption}

\usepackage{multicol}
\usepackage{amsfonts}
\usepackage{booktabs}

\usepackage{multirow}

\usepackage{threeparttable}
\usepackage[usenames,dvipsnames]{color}
\usepackage{colortbl}

\usepackage{siunitx}
\newcommand{\ignore}[1]{{}}

\makeatletter

\newcommand{\Rmnum}[1]{\expandafter\@slowromancap\romannumeral #1@}
\makeatother

\usepackage{siunitx}

\begin{document}
% \title{\huge Low-Latency and Fresh Service Provision in Edge Networks via Online Joint Caching, Offloading and Resource Allocation}
% \title{\huge Online Low-Latency and Fresh Service Provision in Mobile Edge Computing Networks via Optimization-Based Deep Reinforcement Learning}
\title{\Large Mobile Edge Computing Networks: Online Low-Latency and Fresh Service Provisioning}

\author{\IEEEauthorblockN
{
Yuhan Yi,~Guanglin Zhang,~and Hai Jiang}
\vspace{-10mm}
\thanks{
% This work was supported in part by the National Natural Science Foundation of China under Grant 62072096,  in part by the International S\&T Cooperation Program of Shanghai Science and Technology Commission under Grant 20220713000,  in part by the Program for Professor of Special Appointment (Eastern Scholar) at Shanghai Institutions of Higher Learning, in part by ``Shuguang Program'' of Shanghai Education Development Foundation and Shanghai Municipal Education Commission, and in part by the Young Top-notch Talent Program in Shanghai.
Y.~Yi and G.~Zhang are with the College of Information Science and Technology, Donghua University, Shanghai 201620, China (e-mail: yuhanyi@mail.dhu.edu.cn; glzhang@dhu.edu.cn.). H.~Jiang is with the Department of Electrical and Computer Engineering, University of Alberta, Edmonton, AB T6G 1H9, Canada (e-mail:
hai1@ualberta.ca).
% Engineering Research Center of Digitized Textile and Apparel Technology, Ministry of Education, Shanghai 201620, China. 
% \textit{Corresponding author: G. Zhang}.
}
}

\maketitle
% \markboth{IEEE}
\markboth{ }
% {G. Zhang \MakeLowercase{\textit{et al.}}: }
\maketitle

\begin{abstract}
Edge service caching can significantly mitigate latency and reduce communication and computing overhead by fetching and initializing services (applications) from clouds. The freshness of cached service data is critical when providing satisfactory services to users, but has been overlooked in existing research efforts. In this paper, we study the online low-latency and fresh service provisioning in mobile edge computing (MEC) networks. Specifically, we jointly optimize the service caching, task offloading, and resource allocation without knowledge of future system information, which is formulated as a joint online long-term optimization problem. This problem is NP-hard. To solve the problem, we design a Lyapunov-based online framework that decouples the problem at temporal level into a series of per-time-slot subproblems. For each subproblem, we propose an online integrated optimization-deep reinforcement learning (OIODRL) method, which contains an optimization stage including a quadratically constrained quadratic program (QCQP) transformation and a semidefinite relaxation (SDR) method, and a learning stage including a deep reinforcement learning (DRL) algorithm. Extensive simulations show that the proposed OIODRL method achieves a near-optimal solution and outperforms other benchmark methods.
\end{abstract}

\begin{IEEEkeywords}
Mobile edge computing, service caching, age of information, deep reinforcement learning
\end{IEEEkeywords}

\IEEEpeerreviewmaketitle

\section{Introduction}\label{Introduction}
% \subsection{Background and Motivation}
Mobile edge computing (MEC) has the potential to achieve low latency services by deploying cloud capabilities to the network edge \cite{10275092}. In MEC, edge servers (ESs) are equipped with communication, computing, and storage resources to cache a variety of services, such as databases and programs, from remote cloud servers (CSs). Hence, the corresponding tasks from users can be processed locally at the ESs, rather than being further offloaded to the CSs, which may incur long latency. 
This technology, referred to as \textit{edge service caching}, can cache popular services, thereby saving repeated backhaul transmissions to and from the CSs.
This alleviates the pressure on backhaul links and reduces system energy consumption. 
Additionally, it can provide low-latency services by processing tasks in close proximity to users \cite{9462377}.
With these advantages, edge service caching benefits newly emerging applications that are computation-intensive and latency-sensitive, including virtual reality, autonomous driving, face recognition, etc. \cite{9599706}. 
%However, due to limited storage and computing resources, ESs can only cache a small set of services and process a portion of tasks from users.
%Therefore, to fully exploit the potential of edge service caching, intelligent service caching, task offloading, computing resource allocation as well as communication resource allocation are required. 

Edge service caching has been extensively studied in the literature, aiming at utilizing limited edge resources to make intelligent service caching decisions that improve system performance.
Many studies focus on the service request delay of users when edge service caching is applied \cite{9983987,10293006,9462377}.
Work \cite{9462377} studies the service caching problem to minimize the request delay in multi-tiered edge cloud networks, and proposes two approximation algorithms for the case of a single service type. An efficient heuristic algorithm considering different service types is also proposed.
The cache storage size of each ES is fixed. Work \cite{9983987} considers adaptive cache storage size and investigates the tradeoff between the service request delay and cache storage cost, since larger cache storage helps to reduce delay while requiring more hardware cost. 
In both works \cite{9462377} and \cite{9983987}, an ES caches a service only when it gets a request of the service from users. On the other hands, work \cite{10293006} predicts future user requests by using a graph attention network-based service request prediction method in intelligent cyber-physical transportation systems, and caches related services in advance.  

The above works \cite{9983987,10293006,9462377} assume that each ES has sufficient computing capability, i.e., the ES is able to process all user requests if the related services are already cached at the ES. However, practically each ES has limited computing capability, which means that the ES may have to \textit{offload} some user requests to the CS even if the related services are cached at the ES. 
Thus, it is necessary to consider task offloading when investigating edge service caching to optimize MEC systems.
Some existing works \cite{9832009,9808348,9362261} jointly consider service caching and task offloading in MEC systems.
Specifically, work \cite{9832009} focuses on the service latency minimization problem, and work \cite{9808348} aims to minimize the system cost including task delay and energy consumption.
Both works \cite{9832009} and \cite{9808348} consider that each user request is allocated a fixed amount of computing and communication resources.
To alleviate users’ competition for limited computing resources, a service pricing scheme is proposed in work \cite{9362261}. All computing resources are equally shared by all tasks served by the ES, while each task has a fixed amount of communication resources.

Generally different tasks have different computational workloads and data sizes, and thus, it is beneficial to adaptively allocate computing and communication resources among users when designing service caching and task offloading.
Accordingly, there are research efforts jointly investigating service caching, task offloading, and resource allocation \cite{9852277,10103637,9815021,9789202} in MEC systems.
Work \cite{9852277} considers a single-server MEC system, using the location-aware users' preference and proposing a hybrid learning framework to solve the joint optimization of service caching, task offloading, and resource allocation. Works \cite{10103637} and \cite{9815021} consider a multi-server collaborative MEC architecture, in which work \cite{10103637} proposes an approximation and decomposition theory-based two-stage algorithm to maximize users' quality of experience, and work \cite{9815021} develops a deep deterministic policy gradient-based algorithm to minimize the energy consumption. Different from works \cite{10103637} and \cite{9815021} with homogeneous ESs, work \cite{9789202} considers heterogeneous ESs, in which an ES could be a base station in cellular networks or an access point in wireless local area networks. Energy consumption is minimized by jointly optimizing service caching, task offloading, and resource allocation.

All the above existing studies \cite{9983987,10293006,9462377,9832009,9808348,9362261,9852277,10103637,9815021,9789202} make efforts to reduce user request delay and/or energy consumption by investigating service caching, task offloading, and/or resource allocation. However, they implicitly assume that the cached edge services (e.g., databases and programs) do not need to be updated/refreshed. In other words, they do not consider the {\it freshness} of cached edge services.
The freshness of service data actually plays a crucial role in the performance of MEC networks. 
For example, in the context of autonomous driving applications, road side units equipped with ESs can cache the neural network model from CSs for objective recognition service \cite{10007043}. The freshness of the neural network model's parameters for processing the road information is quite critical in practice since the road conditions may change over time, which requires a fresh neural network model.
Thus, in edge service caching networks, the freshness of service data should also be considered together with other system performance such as delay.

Our paper aims to fill the research gap of service data freshness at ESs. Specifically, we investigate online low-latency and fresh service provisioning in MEC networks by jointly optimizing service caching, task offloading, and resource allocation. An online low-complexity and near-optimal solution is provided.

The novelty of our research and our technical contributions are presented below:
\begin{itemize}
% \item \textbf{A novel model}: 
% This work is the first to consider the freshness of edge service caching, where the freshness is measured by the AoI.
% We aim to provide online low-latency and fresh services to users in MEC networks at the minimum system cost, and this problem is formulated as an online long-term joint optimization problem. 
% Different from existing studies, we consider the freshness of edge service caching.
% Thus, not only the caching decisions, offloading decisions, and resource allocation decisions, but also the refreshing of cached services need to be considered.
% % Besides, the future system information, including future task requests from users, future data size of service, future updating of service data, and etc., is dynamic and unknown.
% Besides, future task and service information is dynamic and unknown.
% We formulate this problem as an online long-term joint optimization problem. 
\item \textbf{Fresh service provisioning in MEC networks}: 
% This work is the first to consider the freshness of edge service caching, where the freshness is measured by the AoI metric.
This work is the first in the research society that points out the importance of service freshness at the ES and develops an effective algorithm to provide online low-latency and fresh services.
%consider the freshness of edge service caching, and the first to utilize the AoI metric in MEC networks with service caching.
\item \textbf{Low-complexity and near-optimal solution}:
The considered problem, i.e., joint design of service caching, task offloading, and resource allocation with consideration of service freshness, is challenging due to the following reasons. First, downloading decision, caching decision,\footnote{The downloading decision refers to whether the ES downloads service data from the CS, while the caching decision refers to whether service data is cached at the ES.} offloading decision, and computing and communication resource allocation decision are coupled together, which makes the research problem extremely hard to solve. Second, the freshness consideration of service data adds temporal coupling of the decisions, making it even harder to achieve optimal performance. To address the second challenge, we use a Lyapunov-based online framework to decouple the temporal correlation among time slots. To address the first challenge, traditional deep reinforcement learning (DRL) methods do not work well for our problem due to dynamic service data size. Thus, we propose to use an optimization method to first provide a rough idea of downloading and caching decisions of services, referred to as the {\it optimization stage},
% (downloading\footnote{In this paper, downloading decisions include the case of service refreshing.}) 
and then use a DRL algorithm to explore a better solution, referred to as the {\it learning stage}.
% which only requires the previous and current system information.
The proposed optimization-plus-DRL method outperforms the optimization-only method and the DRL-only method.
% This optimization plus DRL method outperforms the optimization-only method and DRL-only method. 
Near-optimal performance is achieved by the Lyapunov-based online framework and the optimization-plus-DRL method.

%This framework transforms the freshness constraint of service caching into a queue stability constraint and decouples the long-term optimization problem into a series of per-time-slot subproblems.
%Besides, a theoretical performance guarantee is provided for this decoupling and shows a cost-freshness tradeoff as $[\mathcal{O}(1/V)$, $\mathcal{O}(V)]$.

%\item \textbf{Hybrid optimization methods combined with deep reinforcement learning (DRL)}: 
%However, the temporal-decoupled subproblem is still NP-hard. To this end, we propose an OIODRL method, which combines the advantages of optimization methods and DRL algorithm.
%Specifically, the optimization methods can provide a rough idea of caching (downloading\footnote{In this paper, downloading decisions include the case of service refreshing.}) decisions that DRL struggles with due to the dynamic service data size.
%Besides, given the rough idea of near-optimality, DRL can explore a better solution, where an optimization module is also used to reduce the output dimension in DRL.
% Specifically, in the optimization stage, we utilize quadratically constrained quadratic program (QCQP) transformation and semidefinite relaxation (SDR) method to provide a rough idea of the solution; in the learning stage, given the rough idea from the optimization stage, we leverage a DRL algorithm, i.e., clip-PPO2, to obtain the final decisions, where a Lagrange multiplier method is used to reduce the dimension of output decision in the DRL algorithm.
\item \textbf{Quadratically constrained quadratic program (QCQP) transformation in the optimization stage:} The optimization stage aims at a rough idea of downloading and caching decisions, which is an NP-hard problem. Through a series of math manipulations, we successfully transform the problem into a QCQP problem, which is effectively solved by the help of the semidefinite relaxation (SDR) method.

\item \textbf{A novel method to reduce the action space in the learning stage}: We propose a method, by using an optimization module, to reduce the action space of the DRL by half, which facilitates convergence and enhances training performance.
%\item \textbf{Theoretical analysis and extensive simulations}:
%Theoretical analysis is provided to show the gap between the rough idea and the optimal decisions.
%Experiment results verify the advantage of using clip-PPO2 algorithm in OIODRL, and show that OIODRL outperforms other benchmark methods and performs similarly to the optimal exhaustive search method.

\end{itemize}

The rest of the paper is arranged as follows.
% Section~\ref{Sec:related work} introduces the related work.
Section~\ref{Sec:system model} presents a detailed definition of the system with edge service caching and formulates a joint online long-term cost minimization problem. Section~\ref{Sec:Lyapunov-Based Online Framework} presents the Lyapunov-based online framework, and Section~\ref{Sec:algorithm} details the optimization-plus-DRL method. Section~\ref{Sec:simulation} presents the experiment results. Section~\ref{Sec:conclusion} concludes the paper.

\begin{figure}
    \centering
    \includegraphics[width=0.65\linewidth]{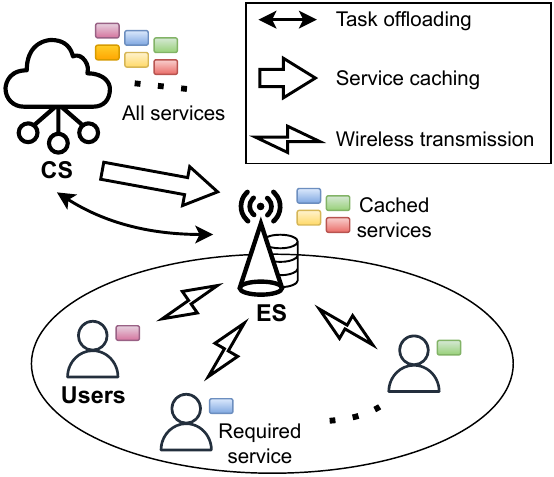}
    \caption{System model illustration.}
    \label{fig:system_model}
\end{figure}

\section{System Model and Problem Formulation}\label{Sec:system model}
Consider a general MEC system, which consists of a CS, a base station (BS) integrated with an ES, and multiple users, as shown in Fig. \ref{fig:system_model}. 
The CS can provide totally $J$ service types denoted as $\mathcal{J}=\{1,2,...,J\}$, e.g., face recognition, virtual/augmented reality, etc.
The CS keeps refreshing the service data for all the service types. 
Service data updating of the $J$ service types are independent and may not occur simultaneously.
For each service type, the updating may not be periodical. 
The BS provides radio access to users and accepts users' service requests.
The ES embedded in the BS can download service data of some service types and cache them locally, which means that the ES is able to process some service requests locally.
Upon the arrival of a user request, the ES can decide to process the request locally or offload it to the CS for processing.
For example, when the ES does not cache the service data or lacks sufficient computing resources to process the request, it may offload the request to the CS.
There are $I$ users denoted as $\mathcal{I}=\{1,2,...,I\}$. We consider $T$ time slots denoted as $\mathcal{T}=\{0,1,...,T-1\}$, and each user can request at most one service type per time slot.\footnote{Our system can be straightforwardly extended to the case where a user could request multiple service types in a time slot. In such a case, each user could be viewed as multiple virtual users located at the same site. Our method can then be applied to serve these multiple virtual users effectively.}

\subsection{Offloading, Downloading, and Caching decisions}
For presentation simplicity, when we say ``service $j\in\mathcal{J}$", it means service data of service type $j$. And if a user requests service type $j$, we say the user has ``task $j$".

At the CS, service type $j$ at time slot $t\in\mathcal{T}$ is represented as a triplet $\zeta_j(t)\triangleq\{S_j(t),p^\text{p}_j(t),p^\text{r}_j(t)\}$: $S_j(t)$ is the data size of service $j$ at time slot $t$, $p^\text{p}_j(t)$ is the \textit{purchasing price} of service $j$ (i.e., the price that the ES should pay the CS for downloading entire service $j$), and $p^\text{r}_j(t)$ is the \textit{refreshing price} of service $j$ (i.e., the price that the ES should pay the CS for downloading the refreshed portion of service $j$). Generally we have $p^\text{r}_j(t)\leq p^\text{p}_j(t)$.

%the purchasing price and refreshing price of service $j$ that the CS sells to the ES at time slot $t$. In other words, $p^\text{d}_j(t)$ and $p^\text{d}_j(t)$ can be regarded as the caching cost and refreshing cost that ES downloads the service data from the CS. As refreshed service data will not be larger than the full service data, we assume that $p^\text{d}_j(t)\leq p^\text{d}_j(t)$. (? Assume that they are related to data size)

%For simplicity, if a user requests service type $j$, we say the user has a task $j$.
For user $i\in\mathcal{I}$, its task $j$ at time slot $t$ is described as a triplet $\xi_{ij}(t)\triangleq\{S_{ij}^\text{u}(t),S_{ij}^\text{d}(t),Y_{ij}(t)\}$, where $S_{ij}^\text{u}(t)$ and $S_{ij}^\text{d}(t)$ are task $j$'s data sizes of wireless uplink (data to be processed) and wireless downlink (processing result) between user $i$ and the ES at time slot $t$,
% \footnote{$S_{ij}^\text{u}(t)$ is the size of the data to be processed, and $S_{ij}^\text{d}(t)$ is the size of the processing result.}
and $Y_{ij}(t)$ means the required CPU cycles when processing user $i$'s task $j$ at time slot $t$. If user $i$ does not request task $j$ at time slot $t$, we have $S_{ij}^\text{u}(t)=0, S_{ij}^\text{d}(t)=0, Y_{ij}(t)=0$.

Each user sends its task to the ES by using the wireless uplink. The ES could process the task locally (e.g., when the ES has the cached service data and has sufficient computing resources), or offload the task to the CS for processing.
Let $x_{ij}(t)\in\{0,1\}$ denote the \textit{offloading decision} of user $i$'s task $j$ at time slot $t$: $x_{ij}(t)=1$ means that the task is processed at the ES locally; and $x_{ij}(t)=0$ means that the task is offloaded to the CS. When the task is offloaded to the CS, the ES should pay the CS for \textit{computing cost}, denoted $c_{ij}(t)$, which is related to the size of the task data (i.e., $S_{ij}^\text{u}(t)$) as
\begin{equation}
c_{ij}(t)=\lambda_s S^\text{u}_{ij}(t)(1-x_{ij}(t)), \label{eq:computing_cost}
\end{equation}
where $\lambda_s$ is the computing cost per unit size.

The ES may download the service data of some services from the CS and cache them locally. Denote $y_j(t)\in\{0,1\}$ as the \textit{downloading decision} of service $j$ at the ES at time slot $t$: 
\begin{itemize} 
\item $y_{j}(t)=1$ means that the ES decides to download service $j$ from the CS at time slot $t$. Specifically, if the service data is cached at the ES at the previous time slot (i.e., time slot $t-1$), the ES only downloads the refreshed portion of the service data from the CS and pays the refreshing price $p^\text{r}_j(t)$; otherwise, the ES downloads the entire service data and pays the purchasing price $p^\text{p}_j(t)$.  
\item $y_{j}(t)=0$ means that the ES decides not to download service $j$ from the CS at time slot $t$.
\end{itemize}  

Denote $z_{j}(t)\in\{0,1\}$ as the \textit{caching decision} of service $j$ at the ES at time slot $t$, indicating that the service is cached at the ES (i.e., $z_{j}(t)=1$) or not (i.e., $z_{j}(t)=0$). In general, when the ES downloads the service data at time slot $t$ (i.e., $y_{j}(t)=1$), then after the downloading, the service data is cached at time slot $t$ (i.e., $z_{j}(t)=1$). In subsequent time slots, say time slot $t'$, the ES may continue to cache the service data (i.e., $z_{j}(t')=1$) or remove the service data from its storage (i.e.,  $z_{j}(t')=0$).

As the storage capacity of the ES is limited, the ES can only cache a partial number of services, which satisfies the storage capacity constraint as follows:
\begin{align}
\sum_{j\in \mathcal{J}} S_j(t)z_j(t)\leq S, \forall t,
\label{cons:S}
\end{align}
where $S$ is the total storage capacity of the ES.

If the ES downloads service $j$ at time slot $t$, denote $p_j(t)$ as the price (purchasing price or refreshing price) that the ES pays to the CS. Then we have% which is related to the previous cache status. In detail, if service $j$ has been cached, the ES only needs to pay the refresh price $p^\text{d}_j(t)$; otherwise, the ES should pay the purchase price $ p_j^\text{d}(t)$.
\begin{align}\label{eq:price}
p_j(t)=p_j^\text{p}(t)(1-z_{j}(t-1))+p_j^\text{r}(t)z_{j}(t-1).
\end{align}

The decisions $x_{j}(t)$, $y_{j}(t)$, and $z_{j}(t)$ are coupled, as follows.
Consider user $i$ and service $j$. 
\begin{itemize} 
\item If service $j$ is not cached at time slot $t-1$ (i.e., $z_j(t-1)=0$), then at time slot $t$ we have: the service data is cached (i.e., $z_j(t)=1$) if the ES downloads the service data from the CS (i.e., $y_j(t)=1$); or not cached (i.e., $z_j(t)=0$) otherwise (i.e., $y_j(t)=0$). In other words, we have
\begin{align}
\text{~if~}z_{j}(t-1)=0,\text{~then~}y_{j}(t)=z_{j}(t),\forall j,t. \label{cons:zy1}
\end{align}

\item If service $j$ is cached at time slot $t-1$ (i.e., $z_j(t-1)=1$), then at time slot $t$ the ES may: (i) download and cache the refreshed portion from the CS (i.e., $y_j(t)=1$ and $z_j(t)=1$); (ii) not change the cached service data (i.e., $y_j(t)=0$ and $z_j(t)=1$); or (iii) remove the service data from its storage (i.e.,  $y_j(t)=0$ and $z_j(t)=0$). In other words, we have 
\begin{align}
\text{~if~}z_{j}(t-1)=1,\text{~then~}y_{j}(t)\leq z_{j}(t),\forall j,t. \label{cons:zy2}
\end{align}

\item Further, since the ES can process a task only when it has cached the related service data, we have
\begin{equation}
x_{ij}(t)\leq z_{j}(t), \forall i,j,t.
\label{cons:xy}
\end{equation}
\end{itemize}

\subsection{Delay Model}
For user $i$'s task $j$, we discuss the delay in two cases: when the task is processed locally at the ES, and when it is offloaded to the CS for processing.
\subsubsection{Processing locally at the ES}
If user $i$'s task $j$ at time slot $t$ is processed locally at the ES, the delay is denoted as $D_{ij}^\text{L}(t)$, which consists of the uplink transmission delay denoted $D_{ij}^\text{u}(t)$, the processing delay of the ES denoted  $D_{ij}^\text{e}(t)$, and the downlink transmission delay for computation result feedback denoted $D_{ij}^\text{d}(t)$. 
Let $w^\text{u}_{ij}(t)$ and $w^\text{d}_{ij}(t)$ denote the uplink and downlink bandwidth allocated to user $i$'s task $j$ at time slot $t$, respectively, which satisfy the following constraints:
\begin{align}
\sum_{i\in \mathcal{I}}\sum_{j\in\mathcal{J}}w^\text{u}_{ij}(t)\leq W^\text{u},~\sum_{i\in \mathcal{I}}\sum_{j\in\mathcal{J}}w^\text{d}_{ij}(t)\leq W^\text{d},\forall t,
\label{cons:Wud}
\end{align}
where $W^\text{u}$ and $W^\text{d}$ are the total uplink and downlink bandwidth, respectively. %denoted by $W^\text{u}$ and $W^\text{d}$, as follows:
%Considering the situation where uplink and downlink transmissions overlap, we have an additional total bandwidth limitation,
%\begin{align}
%\sum_{i\in \mathcal{I}}\sum_{j\in\mathcal{J}}w^\text{d}_{ij}(t)+ w^\text{d}_{ij}(t)\leq W,\forall t,
%\label{cons:W}
%\end{align}
%where $W$ is the total bandwidth capacity and satisfies $W\leq W^\text{d}+W^\text{d}$.
Then, transmission delays $D_{ij}^\text{u}(t)$ and $D_{ij}^\text{d}(t)$ are calculated by $D_{ij}^\text{u}(t)= \frac{S_{ij}^\text{u}(t)}{\eta^\text{u}_i(t) w^\text{u}_{ij}(t)}$ and
$D_{ij}^\text{d}(t)=\frac{S_{ij}^\text{d}(t)}{\eta^\text{d}_i(t) w^\text{d}_{ij}(t)}$, where $\eta^\text{u}_i(t)$ and $\eta^\text{d}_i(t)$ are the uplink and downlink spectral efficiency (the achievable transmission rate over 1Hz bandwidth) at time slot $t$, respectively.
%Note that $\eta^\text{u}_i(t)$ and $\eta^\text{d}_i(t)$ are dynamic and vary with time, reflecting the wireless channel states.
Besides, denote $f^\text{e}_{ij}(t)$ as the computing rate (the number of CPU cycles per unit time) allocated to user $i$'s task $j$ at time slot $t$, which is constrained by the ES's total computing capacity denoted $F$ as follows:
\begin{align}
\sum_{i\in \mathcal{I}}\sum_{j\in \mathcal{J}}f^\text{e}_{ij}(t)\leq F, \forall t.
\label{cons:F}
\end{align}
Then the processing delay at the ES is $D_{ij}^\text{e}(t) = \frac{Y_{ij}(t)}{f^\text{e}_{ij}(t)}$.

Thus, we have the total delay of the local processing as
\begin{align}\label{eq:delay_ES}
D_{ij}^\text{L}(t) = D_{ij}^\text{u}(t)+D_{ij}^\text{e}(t)+D_{ij}^\text{d}(t).
\end{align}
\subsubsection{Offloading to the CS}
If user $i$'s task $j$ at time slot $t$ is offloaded to the CS for processing, the delay is denoted as $D_{ij}^\text{O}(t)$, which includes the uplink transmission delay $D_{ij}^\text{u}(t)$, the transmission delay between the ES and CS denoted $D_{ij}^\text{ec}(t)$, the cloud processing delay denoted $D_{ij}^\text{c}(t)$, and the downlink transmission delay for computation result feedback $D_{ij}^\text{d}(t)$, as follows:
\begin{align}\label{eq:delay_CS}
D_{ij}^\text{O}(t) = D_{ij}^\text{u}(t)+D_{ij}^\text{ec}(t)+D_{ij}^\text{c}(t)+D_{ij}^\text{d}(t),
\end{align}
where $D_{ij}^\text{u}(t)$ and $D_{ij}^\text{d}(t)$ have the same expressions as those in the case of processing locally at the ES.
$D_{ij}^\text{ec}(t)$ is expressed as
$D_{ij}^\text{ec}(t) = \frac{S_{ij}^\text{u}(t)+S_{ij}^\text{d}(t)}{r^\text{ec}}$, where $r^\text{ec}$ represents the transmission rate between the ES and CS.
%The reason behind it is the ES-CS link, which is commonly a wired connection with high capacity, in contrast to the limited wireless user-ES link. Consequently, the need for considering bandwidth sharing among users is eliminated.
$D_{ij}^\text{c}(t)$ is expressed as $D_{ij}^\text{c}(t) = \frac{Y_{ij}(t)}{f^\text{c}}$, where $f^\text{c}$ is the computing rate allocated by the CS to each task. %$f^\text{c}$ is a pre-determined constant considering the substantial computing capacity of the CS. 

Overall, the processing delay of user $i$'s task $j$ at time slot $t$ is denoted as $D_{ij}(t)$ and given by 
\begin{align}
D_{ij}(t)=D_{ij}^\text{L}(t)x_{ij}(t)+D_{ij}^\text{O}(t)(1-x_{ij}(t)). \label{eq:delay_cost}
\end{align}

\subsection{Age of Information Model}
We leverage the age of information (AoI) metric to quantify the freshness of service data. For any service data, its AoI is defined as the time elapsed since the moment when the service data version is generated at the CS.

Since both the CS and the ES have service data, we have AoI for each service data at the CS and at the ES, as follows. 

Let $A_{j}^\text{c}(t)$ denote the AoI (unit: time slot) of service $j$ at the CS at the beginning of time slot $t$, which is the time elapsed since the most recent updating of the service data at the CS. %Recall that the updating of service data is independent and may not be periodical.

Let $A_{j}^\text{e}(t)$ denote the AoI of the cached service $j$ at the ES, which is the time elapsed since the cached service version was generated at the CS.
At time slot $t$, when the ES downloads and caches service $j$ from the CS (i.e., $z_j(t)=1,y_j(t)=1$), we have $A_{j}^\text{e}(t) = A_{j}^\text{c}(t)$; when the ES does not download service $j$ from the CS but still keeps caching the cached service $j$ (i.e., $z_j(t)=1,y_j(t)=0$), we have $A_{j}^\text{e}(t)=A_{j}^\text{e}(t-1)+1$; when the ES does not cache service $j$ at time slot $t$ (i.e., $z_j(t)=0,y_j(t)=0$), we virtually define the AoI of service $j$ at the ES the same as that at the CS, i.e. $A_{j}^\text{e}(t) = A_{j}^\text{c}(t)$.\footnote{The reason for this setting is because if service $j$ is not cached at the ES, then the corresponding task will be served at the CS, which means that the service $j$ that the user receives has AoI $A_{j}^\text{c}(t)$.}

%since the last service updating at the CS before the most recent caching moment of the service. Here the caching moment means the moment when the ES downloads and caches service data from the CS. It can be seen that $A_{j}^\text{e}(t)$ is meaningful only when service $j$ is cached at the ES (i.e., $z_j(t)=1$). Upon each downloading of service $j$ from the CS (i.e., $y_j(t)=1$), the AoI of service $j$ at the ES is synchronized to the AoI of service $j$ at the CS. 
% Specifically, if the ES downloads service $j$ at time slot $t$ (i.e., $y_j(t)=1$), its AoI of service $j$ would become $A_{j}^\text{c}(t)$. %, where $\alpha_j^\text{e}(t)$ refers to the update transmission delay from CS to ES and is also assumed to be much smaller than $\tau$; (for the same reason, $\alpha_j^\text{e}(t)$ also in milliseconds?)
So the AoI of service $j$ at the ES at time slot $t$ is given by
\begin{equation}\label{eq:AoI_edge}
A_{j}^\text{e}(t)=
\begin{cases}
A_{j}^\text{c}(t),&\text{~if~}z_j(t)=1,~y_j(t)=1,\\
A_{j}^\text{e}(t-1)+1,&\text{~if~}z_j(t)=1,~y_j(t)=0,\\
A_{j}^\text{c}(t),&\text{~if~}z_j(t)=0,~y_j(t)=0.
\end{cases}
\end{equation}

% Note that the service of the CS is always fresher than that of the ES, i.e., $A_{j}^\text{c}(t)\leq A_{j}^\text{e}(t)$. 
% %since the data update of service always arrives at the CS and then is transferred to the ES.
% If service $j$ is not cached at the ES (i.e., $z_j(t)=0$), we virtually define the AoI of service $j$ at the ES as that at the CS\footnote{The reason for this setting is because if service $j$ is not cached at the ES, then the corresponding task will be served at the CS.}. 
% Thus, the AoI of service $j$ at the ES considering this case is denoted as $A_{j}(t)$ and given by
% % Thus, the AoI of service $j$ is expressed as
% \begin{align}\label{eq:AoI}
% A_{j}(t) = A_{j}^\text{c}(t)(1-z_{j}(t))+A_{j}^\text{e}(t)z_{j}(t).
% \end{align}

\subsection{Problem Formulation}\label{Problem Formulation}
We consider the total cost of the ES as a weighted sum of delay cost, computing cost, and service purchasing/refreshing cost, which is denoted as $C(t)$ and given by
\begin{align}\label{e:ct_expression}
C(t)= \sum_{i\in \mathcal{I}}\sum_{j\in \mathcal{J}}[\lambda_D D_{ij}(t)+\lambda_c c_{ij}(t)]+\sum_{j\in \mathcal{J}}\lambda_p p_{j}(t)y_{j}(t),
\end{align}
where $D_{ij}(t)$ is given in (\ref{eq:delay_cost}), $c_{ij}(t)$ is given in (\ref{eq:computing_cost}), and $p_{j}(t)$ is given in (\ref{eq:price}).
$\lambda_D$, $\lambda_c$, and $\lambda_p$ are weights of the three costs.
% $D(t)$ is the sum of delay of all users' tasks as $D(t)\triangleq\sum_{i\in \mathcal{I}}\sum_{j\in \mathcal{J}}D_{ij}(t)$, $c(t)$ is the sum of computing cost for all users' tasks as $c(t)\triangleq\sum_{i\in \mathcal{I}}\sum_{j\in \mathcal{J}}c_{ij}(t)$, $p(t)$ is the sum of all service purchasing/refreshing cost $p(t)\triangleq\sum_{j\in \mathcal{J}}p_{j}(t)y_{j}(t)$ ($p_j(t)$ is given in (\ref{eq:price})). 
%; $\lambda_D$, $\lambda_c$, and $\lambda_p$ are the relative weights of the above three costs, by which we can attach different importance to them.

Besides providing low-latency services for users, the ES also ensures that its services are fresh to serve users' tasks. Therefore, it is required that the long-term average AoI of each service at the ES is bounded by a threshold value. Denote the threshold value for service type $j$ as $A_j^{\max}$, and we have 
\begin{align}
\lim\limits_{T\to \infty}\frac{1}{T}\sum_{t=0}^{T-1}\mathbb{E}[A_j^\text{e}(t)]\leq A_j^{\max},~\forall j. \label{cons:aoi}
\end{align}

To minimize the long-term average cost of the ES, we formulate the following optimization problem:
\begin{subequations}
\begin{align}
\mathcal{P}1:~
&\min_{\boldsymbol x,\boldsymbol y,\boldsymbol z,\boldsymbol{\omega},\boldsymbol f}~\lim\limits_{T\to \infty}\frac{1}{T}\sum_{t=0}^{T-1} \mathbb{E}[C(t)]\nonumber\\
\text{s.t.}~~
&(\ref{cons:S}),(\ref{cons:zy1}),(\ref{cons:zy2}),(\ref{cons:xy}),(\ref{cons:Wud}),(\ref{cons:F}),(\ref{cons:aoi}),\notag \\
&x_{ij}(t),y_{j}(t),z_{j}(t) \in \{0,1\},~\forall i,j,t,\label{cons:xyz}\\
&f^\text{e}_{ij}(t)\geq 0,~\forall i,j,t,\label{cons:fij}\\
&w^\text{d}_{ij}(t),w^\text{d}_{ij}(t)\geq 0,~\forall i,j,t,\label{cons:wij}
\end{align}
\end{subequations}
where
$\boldsymbol x\triangleq [x_{ij}(t)]_{i\in\mathcal{I},j\in\mathcal{J},t\in\mathcal{T}}$,
$\boldsymbol y\triangleq [y_{j}(t)]_{j\in\mathcal{J},t\in\mathcal{T}}$,
$\boldsymbol z\triangleq [z_{j}(t)]_{j\in\mathcal{J},t\in\mathcal{T}}$,
$\boldsymbol{\omega}\triangleq [w_{ij}^\text{u}(t),w_{ij}^\text{d}(t)]_{i\in\mathcal{I},j\in\mathcal{J},t\in\mathcal{T}}$, and
$\boldsymbol f\triangleq [f_{ij}^\text{e}(t)]_{i\in\mathcal{I},j\in\mathcal{J},t\in\mathcal{T}}$.
%The objective of $\mathcal{P}1$ is to minimize the long-term time average cost including delay cost, computing cost, and caching cost, while keeping the service fresh. 
%The constraint (\ref{cons:xyz}) illustrates the $x_{ij}(t)$, $y_{j}(t)$, and $z_{j}(t)$ are binary decision variables; the constraints (\ref{cons:fij}) and (\ref{cons:wij}) show that $f^\text{e}_{ij}(t)$, $w^\text{u}_{ij}(t)$, and $w^\text{d}_{ij}(t)$ are non-negative decision variables.

Problem $\mathcal{P}1$ is a non-convex optimization problem due to the presence of binary decisions $x_{ij}(t)$, $y_{j}(t)$, and $z_{j}(t)$. $x_{ij}(t)$, $y_{j}(t)$, and $z_{j}(t)$ are coupled within each time slot $t$ as shown in  (\ref{cons:S}), (\ref{cons:zy1}), (\ref{cons:zy2}), and (\ref{cons:xy}). They are also temporally coupled as shown in (\ref{eq:price}) and (\ref{cons:aoi}). %Moreover, they are coupled at the temporal level through recursive AoI equation (\ref{cons:aoi}) and service purchasing/refreshing cost (\ref{eq:price}) associated with cache status.
Problem $\mathcal{P}1$ is NP-hard, as illustrated below. Consider a simplified version of problem $\mathcal{P}1$, where $p_j(t)=0$ and constraint (\ref{cons:aoi}) is ignored. Then the simplified problem $\mathcal{P}1$ reduces to minimizing cost $C(t)$ at each time slot. Minimizing cost $C(t)$ at a time slot is a mixed-integer nonlinear programming problem, which is NP-hard. Consequently, the original problem $\mathcal{P}1$ is also NP-hard.

\section{Problem Decoupling with Lyapunov-Based Online Framework}\label{Sec:Lyapunov-Based Online Framework}
% \begin{figure*}
%     \centering
%     \includegraphics[width=0.8\linewidth]{fig/system-pipeline.pdf}
%     \caption{The framework of OJ-CORA algorithm for low-latency and fresh service provision in edge networks.}
%     \label{fig:system_pipelinel}
% \end{figure*}

% \begin{figure*}
%     \centering
%     \includegraphics[width=0.9\linewidth]{fig/algorithm.pdf}
%     \caption{The structure of optimization based PPO2.}
% \end{figure*}

In this section, to solve the challenging optimization problem $\mathcal{P}1$,
we first simplify problem $\mathcal{P}1$ and decouple it into two independent subproblems.
The first subproblem is convex and can be solved via Lagrange multiplier method.
The second subproblem is NP-hard.
Then we utilize Lyapunov optimization technology to decouple the NP-hard problem at temporal level into a series of per-time-slot subproblems.

% Then we propose an Online Joint Caching, Offloading and Resource Allocation (OJ-CORA) algorithm to solve this NP-hard problem.
% The idea of OJ-CORA is to first tackle the difficulty of temporal coupling by utilizing Lyapunov optimization theory to decouple the NP-hard problem along the time slot into a series of subproblems.
% It is worth mentioning that OJ-CORA is an online algorithm without the requirement of dynamic future information (e.g., future task, wireless channel state, purchase and refresh price) and can obtain the approximate optimal solution with a low complexity based only on the current and historical system information.

\subsection{Problem Simplification}
From (\ref{e:ct_expression}), $C(t)$ can be expressed as
% $C(t)=\sum_{i\in \mathcal{I}}\sum_{j\in \mathcal{J}}\{\lambda_D [D_{ij}^\text{L}(t)x_{ij}(t)+D_{ij}^\text{O}(t)(1-x_{ij}(t))]+\lambda_c \lambda_s S^\text{u}_{ij}(t)(1-x_{ij}(t))\}+\sum_{j\in \mathcal{J}}\lambda_p p_{j}(t)y_{j}(t) = C'(t)-U(t)$,
% {\small
\begin{align}
C(t)=&\sum_{i\in \mathcal{I}}\sum_{j\in \mathcal{J}}\{\lambda_D [D_{ij}^\text{L}(t)x_{ij}(t)+D_{ij}^\text{O}(t)(1-x_{ij}(t))]\nonumber\\
&+\lambda_c \lambda_s S^\text{u}_{ij}(t)(1-x_{ij}(t))\}+\sum_{j\in \mathcal{J}}\lambda_p p_{j}(t)y_{j}(t)\nonumber\\
=& C'(t)-U(t),\nonumber
\end{align}
% }
where
% $C'(t)=\sum_{i\in\mathcal{I}}\sum_{j\in\mathcal{J}}[\lambda_D D_{ij}^\text{O}(t)+\lambda_c\lambda_s S^\text{d}_{ij}(t)]$
% {\small
\begin{align}
C'(t)=\sum_{i\in\mathcal{I}}\sum_{j\in\mathcal{J}}[\lambda_D D_{ij}^\text{O}(t)+\lambda_c\lambda_s S^\text{d}_{ij}(t)]\nonumber
\end{align}
% }
with $D_{ij}^\text{O}(t)=\frac{S_{ij}^\text{u}(t)}{\eta^\text{u}_i(t) w^\text{u}_{ij}(t)}+D_{ij}^\text{ec}(t)+D_{ij}^\text{c}(t)+\frac{S_{ij}^\text{d}(t)}{\eta^\text{d}_i(t) w^\text{d}_{ij}(t)}$ from (\ref{eq:delay_CS}), and $U(t)$ is the {\it system utility}, expressed as
% $U(t)=\sum_{i\in\mathcal{I}}\sum_{j\in\mathcal{J}} [\lambda_D (D_{ij}^\text{O}(t) - D_{ij}^\text{L}(t))x_{ij}(t)+\lambda_c \lambda_s S^\text{u}_{ij}(t)x_{ij}(t)]- \sum_{j\in \mathcal{J}}\lambda_pp_{j}(t)y_{j}(t)=\sum_{i\in\mathcal{I}}\sum_{j\in\mathcal{J}} [\lambda_D(D_{ij}^\text{ec}(t)+D_{ij}^\text{c}(t)-\frac{Y_{ij}(t)}{f^\text{e}_{ij}(t)})x_{ij}(t)+\lambda_c \lambda_s S^\text{u}_{ij}(t)x_{ij}(t)]-\sum_{j\in \mathcal{J}}\lambda_p p_{j}(t)y_{j}(t)$
% {\small
\begin{align}
U(t)=&\sum_{i\in\mathcal{I}}\sum_{j\in\mathcal{J}} [\lambda_D (D_{ij}^\text{O}(t) - D_{ij}^\text{L}(t))x_{ij}(t)\nonumber\\
&+\lambda_c \lambda_s S^\text{u}_{ij}(t)x_{ij}(t)]-\sum_{j\in \mathcal{J}}\lambda_p p_{j}(t)y_{j}(t)\nonumber\\
=&\sum_{i\in\mathcal{I}}\sum_{j\in\mathcal{J}} [\lambda_D(D_{ij}^\text{ec}(t)+D_{ij}^\text{c}(t)-\frac{Y_{ij}(t)}{f^\text{e}_{ij}(t)})x_{ij}(t)\nonumber\\
&+\lambda_c \lambda_s S^\text{u}_{ij}(t)x_{ij}(t)]-\sum_{j\in \mathcal{J}}\lambda_p p_{j}(t)y_{j}(t)\nonumber
\end{align}
% }
with the second equality coming from (\ref{eq:delay_ES}) and (\ref{eq:delay_CS}).
It is seen that $C'(t)$ is a function of $w_{ij}^\text{u}(t)$ and $w_{ij}^\text{d}(t)$, and is independent of $x_{ij}(t)$, $y_j(t)$, $z_j(t)$, and $f^\text{e}_{ij}(t)$, while $U(t)$ is a function of  $x_{ij}(t)$, $y_j(t)$, $z_j(t)$, and $f^\text{e}_{ij}(t)$, and is independent of $w_{ij}^\text{u}(t)$ and $w_{ij}^\text{d}(t)$. Thus, problem $\mathcal{P}1$ can be simplified into two subproblems as follows:
\begin{align}
\mathcal{P}1(a):~
&\min_{\boldsymbol{\omega}}~\lim\limits_{T\to \infty}\frac{1}{T}\sum_{t=0}^{T-1} \mathbb{E}[C'(t)]\nonumber\\
\text{s.t.}~~
&(\ref{cons:Wud}),(\ref{cons:wij}).\notag
\end{align}
\begin{align}
\mathcal{P}1(b):~
&\max_{\boldsymbol x,\boldsymbol y, \boldsymbol z, \boldsymbol f}~\lim\limits_{T\to \infty}\frac{1}{T}\sum_{t=0}^{T-1} \mathbb{E}[U(t)] \nonumber\\
\text{s.t.}~~
&(\ref{cons:S}),(\ref{cons:zy1}),(\ref{cons:zy2}),(\ref{cons:xy}),(\ref{cons:F}),(\ref{cons:aoi}),(\ref{cons:xyz}),(\ref{cons:fij}).\notag
\end{align}

Problem $\mathcal{P}1(a)$ is only related to bandwidth resource allocation and is not temporally coupled. 
Therefore, problem $\mathcal{P}1(a)$ is equivalent to minimizing $C'(t)$ at each time slot, shown in the following (by ignoring terms $D_{ij}^\text{ec}(t)$, $D_{ij}^\text{c}(t)$, and $\lambda_c\lambda_s S^\text{d}_{ij}(t)$ since they are independent of $w_{ij}^\text{u}(t)$ and $w_{ij}^\text{d}(t)$):
\begin{align}
\min_{\boldsymbol{\omega}(t)}~&\sum_{i\in \mathcal{I}}\sum_{j\in\mathcal{J}} \frac{S_{ij}^\text{u}(t)}{\eta^\text{u}_i(t) w^\text{u}_{ij}(t)}+\frac{S_{ij}^\text{d}(t)}{\eta^\text{d}_i(t) w^\text{d}_{ij}(t)}\label{problem:P1(a)}\\
\text{s.t.}~~
&(\ref{cons:Wud}),(\ref{cons:wij}),\notag
\end{align}
where $\boldsymbol{\omega}(t)\triangleq [w_{ij}^\text{u}(t),w_{ij}^\text{d}(t)]_{i\in\mathcal{I},j\in\mathcal{J}}$.
The problem in (\ref{problem:P1(a)}) is convex. Based on Lagrange multiplier method, the solutions to (\ref{problem:P1(a)}) and problem $\mathcal{P}1(a)$ can be obtained in closed form as follows: 
\begin{align}
w^\text{u}_{ij}(t)&=\frac{\sqrt{S_{ij}^\text{u}(t)/\eta^\text{u}_i(t)}}{\sum_{i\in\mathcal{I}}\sum_{j\in\mathcal{J}}\sqrt{S_{ij}^\text{u}(t)/\eta^\text{u}_i(t)}}W^\text{u},\label{sol_wu}\\
w^\text{d}_{ij}(t)&=\frac{\sqrt{S_{ij}^\text{d}(t)/\eta^\text{d}_i(t)}}{\sum_{i\in\mathcal{I}}\sum_{j\in\mathcal{J}}\sqrt{S_{ij}^\text{d}(t)/\eta^\text{d}_i(t)}}W^\text{d}.\label{sol_wd}
\end{align}

Problem $\mathcal{P}1(b)$ is still non-convex and NP-hard, to be investigated in the subsequent subsection.

% In other words, the optimal $\bf f$ for problem $\mathcal{P}1(b)$ is a function of $\boldsymbol x,\boldsymbol y, \boldsymbol z$.   
% However, offloading, caching, and downloading decisions couple with each others and coupled temporally, which leads to problem $\mathcal{P}1(b)$ non-convex and NP-hard.
% Besides, the solution to $\mathcal{P}1(b)$ implies that the maximum delay performance improvement is achieved at the lowest cost over a long time horizon while guaranteeing a small service freshness.

\subsection{Lyapunov-Based Online Framework for Problem $\mathcal{P}1(b)$}\label{Sec:Lyapunov}
We aim to present an online solution for problem $\mathcal{P}1(b)$. To this end, we resort to Lyapunov optimization technology \cite{Lyapunov}. First, the long-term average AoI constraint (\ref{cons:aoi}) is transformed into a queue stability constraint.
Specifically, for each service type $j\in \mathcal{J}$, we construct a virtual queue. %, and the queue stability indicates the feasibility of the constraint (\ref{cons:aoi}). 
The input of the virtual queue at time slot $t$ is $A^\text{e}_j(t)$ and the output of the virtual queue at every time slot is $A_j^{\max}$.
Denote $Q_j(t)$ as the \textit{queue backlog} at time slot $t$. Then, the virtual queue is updated at time slot $t+1$ as follows:
\begin{align}\label{eq:Q}
Q_j(t+1)=\max[Q_j(t)-A_j^{\max}+A^\text{e}_j(t),0],~\forall j.
\end{align}
It is seen from \cite{Lyapunov} that if the above virtual queue $Q_j(t)$ is {\it strong stable},\footnote{The queue $Q_j(t),~j\in\mathcal{J}$ is strong stable if
$\lim\limits_{T\to \infty}\frac{1}{T}\sum_{t=0}^{T-1}\mathbb{E}[Q_j(t)]<\infty,~\forall j$ \cite{Lyapunov} which indicates the long-term average backlog of $Q_j(t)$ is finite.\label{footnote:stable}} constraint (\ref{cons:aoi}) for service type $j$ is satisfied. 
% Besides, the queue backlogs serves as sufficient statistics to provide the basis for the system decisions, which allows the algorithm not to hold dynamic future information.
Therefore, problem $\mathcal{P}1(b)$ is equivalent to maximizing the utility $\lim\limits_{T\to \infty}\frac{1}{T}\sum_{t=0}^{T-1} \mathbb{E}[U(t)]$, subject to strong stable virtual queues for $j\in\mathcal{J}$ and constraints (\ref{cons:S}), (\ref{cons:zy1}), (\ref{cons:zy2}), (\ref{cons:xy}), (\ref{cons:F}), (\ref{cons:xyz}), (\ref{cons:fij}). 
However, the equivalent problem is still temporally coupled, and thus, more math manipulations based on Lyapunov optimization technology are needed, as follows.

To stabilize the virtual queue, we construct the \textit{Lyapunov function} and \textit{conditional Lyapunov drift}.

\textit{Lyapunov Function.} 
Define the Lyapunov function at time slot $t$ as 
\begin{align}\label{eq:Lyapunov_function}
L(t)=\frac{1}{2}\sum_{j\in\mathcal{J}}Q^2_j(t),
\end{align}
which implies the size of the queue backlogs of all services. A small $L(t)$ indicates that all queues have small backlogs, further suggesting that all services provided are fresh (since a small backlog of $Q_j(t)$ means that the input of the virtual queue is largely less then the output). Therefore, by keeping $L(t)$ small, the freshness of services can be improved.

\textit{Conditional Lyapunov Drift.} 
Denote the conditional Lyapunov drift at time slot $t$ as $\Delta(t)$ and let $\boldsymbol{O}(t)\triangleq [Q_j(t),A_j^\text{c}(t),A_j^\text{e}(t-1)]_{j\in\mathcal{J}}$ denote the \textit{AoI observation} at time slot $t$. The conditional Lyapunov drift is constructed as
% \begin{align}\label{eq:drift}
% \Delta(t)&=\mathbb{E}[L(t+1)-L(t)|\boldsymbol{O}(t)]\nonumber\\
% &=\frac{1}{2}\sum_{j\in\mathcal{J}}\mathbb{E}[Q^2_j(t+1)-Q^2_j(t)|\boldsymbol{O}(t)],
% \end{align}
\begin{align}\label{eq:drift}
\Delta(t)&=\mathbb{E}[L(t+1)-L(t)|\boldsymbol{O}(t)],
\end{align}
which indicates the expected change of Lyapunov function from time slot $t$ to time slot $t+1$ given current AoI observation $\boldsymbol{O}(t)$.
% $\Delta(t)$ depends on the decisions that react to dynamic task arrival, dynamic service information, and current AoI observation.
Small $\Delta(t)$ will contribute to stabilizing the queues and satisfying constraint (\ref{cons:aoi}), so that the service freshness can be kept at a low level. 
% while probably leading to a small system utility.

Next, to consider both queue stability (by keeping $\Delta(t)$ small) and utility maximization (by keeping $\mathbb{E}[U(t)|\boldsymbol{O}(t)]$ large), we construct the \textit{Lyapunov drift-plus-penalty function}.

\textit{Lyapunov Drift-Plus-Penalty Function. } Define the Lyapunov drift-plus-penalty function at time slot $t$ as
\begin{align}\label{eq:Delta_V}
\Delta_V(t)=\Delta(t)-V\mathbb{E}[U(t)|\boldsymbol{O}(t)],
\end{align}
where $V\geq 0$ is the weight parameter used to tune the importance of the utility and a larger value of $V$ means that the utility is more emphasized. The case of $V=0$ means that we only consider the service freshness but ignore the system utility, while $V>0$ means that we consider both.
% \footnote{The situation of $V=0$ means that we only consider the service freshness but ignore the utility, while $V>0$ means that we consider both.}

To maintain the queue stability and maximize utility, we turn to minimize $\Delta_V(t)$ to obtain the solution of problem $\mathcal{P}1(b)$.
% We utilize \textit{strong stable} to describe the stability of a queue, which is defined as follows:
% \begin{definition}[Strong Stable]\label{Def:Q}
% The queue $Q_j(t)$ is strong stable if [\ref{?}]:
% \begin{align}
% \lim\limits_{T\to \infty}\frac{1}{T}\sum_{t=0}^{T-1}\mathbb{E}[Q_j(t)]<\infty,~\forall j.
% \end{align}
% \end{definition}
% From the definition of strong stable given in footnote \ref{footnote:stable}, if the long-term average \textit{backlog} of the virtual queue $Q_j(t)$ is finite, then the queue is strong stable, which means constraint (\ref{cons:aoi}) is satisfied.
However, due to the $\max[\cdot]$ function in the expression of $Q_j(t)$ in (\ref{eq:Q}), it is difficult to minimize $\Delta_V(t)$ straightforwardly.
Thus we turn to minimize an upper bound of $\Delta_V(t)$ as follows.

From the definition of $Q_j(t)$ in (\ref{eq:Q}), it can be derived that
\begin{align}\label{ineq:Q}
Q^2_j(t+1)\leq &Q^2_j(t)+(A_j^{\max})^2+(A^\text{e}_j(t))^2 \nonumber\\
&+2Q_j(t)(A^\text{e}_j(t)-A_j^{\max}),~\forall j.
\end{align}
Then, according to (\ref{eq:Lyapunov_function}), (\ref{eq:drift}), and (\ref{ineq:Q}), we have
\begin{align}\label{ineq:drift}
\Delta(t)\leq&\frac{1}{2}\sum_{j\in\mathcal{J}}\{(A_j^{\max})^2+\mathbb{E}[(A^\text{e}_j(t))^2|\boldsymbol{O}(t)] \nonumber\\
&+2Q_j(t)\mathbb{E}[A^\text{e}_j(t)|\boldsymbol{O}(t)]-2Q_j(t)A_j^{\max}\}.
\end{align}
From (\ref{eq:AoI_edge}), we have
\begin{align}
A^\text{e}_j(t)=&A_j^\text{c}(t)z_j(t)y_j(t)+(A_j^\text{e}(t-1)+1)z_j(t)(1-y_j(t)) \nonumber\\
&+A_j^\text{c}(t)(1-z_j(t))(1-y_j(t)) \nonumber\\
\overset{(\text{i})}{=}&A_j^\text{c}(t)+(A_j^\text{e}(t-1)+1-A_j^\text{c}(t))(z_j(t)-y_j(t)), \nonumber
\end{align}
where step $(\text{i})$ is from the fact that $z_j(t)y_j(t)=y_j(t)$.
Similarly, we have 
\begin{align}
&(A^\text{e}_j(t))^2\nonumber\\
&=(A_j^\text{c}(t))^2+\alpha^2(z_j(t)-y_j(t))^2+2A_j^\text{c}(t)\alpha(z_j(t)-y_j(t))\nonumber\\
&\overset{(\text{ii})}{=}(A_j^\text{c}(t))^2+[(A_j^\text{e}(t-1)+1)^2-(A_j^\text{c}(t))^2](z_j(t)-y_j(t))\nonumber
\end{align}
where $\alpha\triangleq A_j^\text{e}(t-1)+1-A_j^\text{c}(t)$, and step $(\text{ii})$ is from the fact that $y^2_j(t)=y_j(t)$, $z_j(t)y_j(t)=y_j(t)$, and$z^2_j(t)=z_j(t)$.

Substituting the above expressions of $A^\text{e}_j(t)$ and $(A^\text{e}_j(t))^2$ into the right hand side of (\ref{ineq:drift}) and according to (\ref{eq:Delta_V}), we obtain an upper bound of $\Delta_V(t)$ as
\begin{align}\label{ineq:Delta_V}
\Delta_V(t)\leq&\sum_{j\in\mathcal{J}}\mathbb{E}\{H_j(t) (z_j(t)- y_j(t))-\sum_{i\in\mathcal{I}} V[\lambda_D(D_{ij}^\text{ec}(t)\nonumber\\
&+D_{ij}^\text{c}(t)-\frac{Y_{ij}(t)}{f^\text{e}_{ij}(t)})+\lambda_c \lambda_s S^\text{u}_{ij}(t)]x_{ij}(t)\nonumber\\
&+V\lambda_p p_{j}(t)y_{j}(t)|\boldsymbol{O}(t)\}+\sum_{j\in\mathcal{J}}I_j(t),
\end{align}
where $H_j(t)\triangleq\frac{(A_j^\text{e}(t-1)+1)^2-(A_j^\text{c}(t))^2}{2}+Q_j(t)(A_j^\text{e}(t-1)+1-A_j^\text{c}(t))$ and
$I_j(t)\triangleq\frac{(A_j^{\max})^2+(A_j^\text{c}(t))^2}{2}-Q_j(t)(A_j^{\max}-A_j^\text{c}(t))$.

For opportunistically minimizing an expectation \cite{Lyapunov}, minimizing the upper bound of $\Delta_V(t)$ given in (\ref{ineq:Delta_V}) is equivalent to minimizing $\sum_{j\in\mathcal{J}}\{H_j(t) (z_j(t)- y_j(t))-\sum_{i\in\mathcal{I}} V[\lambda_D(D_{ij}^\text{ec}(t)+D_{ij}^\text{c}(t)-\frac{Y_{ij}(t)}{f^\text{e}_{ij}(t)})+\lambda_c \lambda_s S^\text{u}_{ij}(t)]x_{ij}(t)+V\lambda_p p_{j}(t)y_{j}(t)+I_j(t)\}$\footnote{It is the right-side term of (\ref{ineq:Delta_V}) by removing the expectation operation.} given AoI observation $\boldsymbol{O}(t)$ at the beginning of a time slot at each time slot. After removing $I_j(t)$ (since $I_j(t)$ is independent of decision variables $x_{ij}(t)$, $y_{j}(t)$, $z_{j}(t)$, $f_{ij}^e(t)$ in problem $\mathcal{P}1(b)$), the optimization problem at time slot $t$ is given by
% Next, we leverage the concept of opportunistically minimizing an expectation \cite{Lyapunov} to obtain the optimal solution for the minimization of $\Delta_V(t)$'s upper bound given in (\ref{ineq:Delta_V}). 
% Specifically, at each time slot, we observe the current system information (including task arrival, service information, and AoI observation), then omit the non-optimized term $I_j(t)$ and make decisions (including offloading, downloading, caching, and computing resource allocation decisions) to minimize
% \begin{align}
% &\sum_{j\in\mathcal{J}}\{H_j(t)(z_j(t)-y_j(t))-V\lambda_D\Delta D(t)\nonumber\\ &-V\lambda_c\Delta c(t)+V\lambda_p p(t)\}\nonumber\\
% =&\sum_{j\in\mathcal{J}}[H_j(t)z_j(t)+(V\lambda_p p_j(t)-H_j(t))y_{j}(t)]\nonumber\\
% &-\sum_{i\in \mathcal{I}}\sum_{j\in \mathcal{J}}V[\lambda_D(D_{ij}^\text{ec}(t)+D_{ij}^\text{c}(t)-\frac{Y_{ij}(t)}{f^\text{e}_{ij}(t)})\nonumber\\
% &+\lambda_c\gamma S^\text{u}_{ij}(t)]x_{ij}(t).\nonumber
% \end{align}
% On the basis of the above analysis, $\mathcal{P}1(b)$ can be rewritten as a series of one-time slot problems $\mathcal{P}2$ as follows:
\begin{align}
\mathcal{P}2:~
&\min_{\boldsymbol x(t),\boldsymbol y(t), \boldsymbol z(t), \boldsymbol f(t)}~g_1(\boldsymbol z(t))+g_2(\boldsymbol y(t))-g_3(\boldsymbol x(t), \boldsymbol f(t)) \nonumber\\
\text{s.t.}~~
&(\ref{cons:S}), (\ref{cons:zy1}), (\ref{cons:zy2}), (\ref{cons:xy}), (\ref{cons:F}), (\ref{cons:xyz}), (\ref{cons:fij}),\notag
\end{align}
where
\begin{align}
g_1(\boldsymbol z(t))\triangleq&\sum_{j\in\mathcal{J}}H_j(t)z_j(t),\nonumber\\
g_2(\boldsymbol y(t))\triangleq&\sum_{j\in \mathcal{J}}(V\lambda_p p_j(t)-H_j(t))y_{j}(t),\nonumber\\
g_3(\boldsymbol x(t), \boldsymbol f(t))\triangleq&\sum_{i\in \mathcal{I}}\sum_{j\in \mathcal{J}} V[\lambda_D(D_{ij}^\text{ec}(t)+D_{ij}^\text{c}(t)-\frac{Y_{ij}(t)}{f^\text{e}_{ij}(t)})\nonumber\\
&+\lambda_c\lambda_s S^\text{u}_{ij}(t)]x_{ij}(t),\nonumber
\end{align}
and 
$\boldsymbol x(t)\triangleq [x_{ij}(t)]_{i\in\mathcal{I},j\in\mathcal{J}}$,
$\boldsymbol y(t)\triangleq [y_{j}(t)]_{j\in\mathcal{J}}$,
$\boldsymbol z(t)\triangleq [z_{j}(t)]_{j\in\mathcal{J}}$,
$\boldsymbol f(t)\triangleq [f_{ij}^\text{e}(t)]_{i\in\mathcal{I},j\in\mathcal{J}}$.

So now we have decoupled the original problem $\mathcal{P}1(b)$ at temporal level into problem $\mathcal{P}2$.
Worth noting that solving problem $\mathcal{P}2$ does not require the future information  (e.g., future task arrival, future service information), which leads to the online optimization design.

Lastly, we prove that the gap between solutions of problem $\mathcal{P}1(b)$ and problem $\mathcal{P}2$ is bounded, as shown next in Theorem \ref{theorem_tradeoff}.
%the utility-backlog (cost-freshness) tradeoff for the problem decoupling is $[\mathcal{O}(1/V)$, $\mathcal{O}(V)]$ shown in Theorem \ref{theorem_tradeoff}, which indicates that the system utility increases as the queue backlogs grow.
To illustrate below, let $a^{\max}\triangleq \max{A^\text{e}_j(t),~\forall j, t}$ and $B\triangleq\sum_{j\in\mathcal{J}}\frac{(a^{\max})^2+(A_j^{\max})^2}{2}$.

% \begin{definition}[Network-Only Algorithm]\label{Def:policy}
% The network-only algorithms is a class of algorithms that make offloading, downloading, caching, and computing resource allocation decisions at each time slot only according to the current task arrival and service information \cite{Lyapunov}.
% \end{definition}

% \begin{lemma}\label{Lemma}
% If the random variable in the network is a stationary process and i.i.d over all time slots, and the problem is feasible, then there is a network-only algorithm satisfies the following conditions for any $v>0$ \cite{Lyapunov}:
% \begin{align}
% &\mathbb{E}[U'(t)]\geq U^\text{opt}-v,~\forall t, \label{Lem:U}\\
% &\mathbb{E}[A'_j(t)]-v \leq A_j^{\max},~\forall t,j, \label{Lem:A}
% \end{align}
% where $U^\text{opt}$ is the optimal long-term average utility of problem $\mathcal{P}$1(b) (i.e., the maximum long-term average utility over all decisions of problem $\mathcal{P}$1(b)), and $U'(t)$ and $A'_j(t)$ are the system utility at time slot $t$ and service $j$'s AoI at time slot $t$ obtained by a network-only algorithm, respectively. 
% \end{lemma}

\begin{assumption}\label{Assumption}
We assume that problem $\mathcal{P}1(b)$ follows the Slater condition \cite{Lyapunov}, which means that there is a network-only algorithm\footnote{A network-only algorithm makes offloading decision, downloading decision, caching decision, and computing resource allocation decision at each time slot only according to current task arrival and service information.} and variables $U_{\epsilon}\geq 0$ and $\epsilon>0$ satisfy
\begin{align}
&\mathbb{E}[U^\text{n}(t)]=U_{\epsilon}, \label{Ass:U}\\
&\mathbb{E}[A^\text{n}_j(t)]+\epsilon \leq A_j^{\max},~\forall j, \label{Ass:A}
\end{align}
where $U^\text{n}(t)$ and $A^\text{n}_j(t)$ are the system utility and service $j$'s AoI at time slot $t$ obtained by a network-only algorithm. 
\end{assumption}

\begin{theorem}\label{theorem_tradeoff}
Denote the utility obtained by solving problem $\mathcal{P}2$ for time slot $t$ as $U^*(t)$ and let $U^*\triangleq\frac{1}{T}\sum_{t=1}^T\mathbb{E}[U^*(t)]$.
The upper bound of the gap between the optimal long-term average utility of problem $\mathcal{P}$1(b) (denoted $U^\text{opt}$) and $U^*$ is
\begin{align}
U^\text{opt}-U^*\leq\frac{B}{V},
\end{align}
which shows that the gap can be arbitrarily small by choosing the $V$ value.
With the solution of problem $\mathcal{P}2$, denote $Q^*_j(t)$ as the queue backlog for service type $j$ at time slot $t$.
The upper bound on long-term average queue backlogs of all services is
\begin{align}
\lim\limits_{T\to \infty}\frac{1}{T}\sum_{t=0}^{T-1}\sum_{j\in \mathcal{J}}\mathbb{E}[Q^*_j(t)]\leq \frac{B+V (U^{opt}-U_{\epsilon})}{\epsilon},
\end{align}
which shows that the long-term average queue backlogs of all services increase at most linearly as $V$ rises.
\end{theorem}
\begin{proof}
The proof is inspired by the proof of \cite[Theorem 4.2]{Lyapunov} and we omit the proof.
\end{proof}

\section{Online Integrated Optimization-DRL Method for Problem $\mathcal{P}2$}\label{Sec:algorithm}
% \begin{figure}
%     \centering
%     \includegraphics[width=0.95\linewidth]{fig/OIODRL.pdf}
%     \caption{The structure of proposed OIODRL method.}
%     \label{fig:OIODRL}
% \end{figure}

Problem $\mathcal{P}2$ is temporally decoupled but still intractable.
Due to constrains (\ref{cons:zy1}), (\ref{cons:zy2}), and (\ref{cons:xy}), terms $g_1(\boldsymbol z(t))$, $g_2(\boldsymbol y(t))$, and $g_3(\boldsymbol x(t),\boldsymbol f(t))$ in problem $\mathcal{P}2$ are coupled to each other. Moreover, the caching decisions in $g_1(\boldsymbol z(t))$ are coupled due to the storage capacity constraint (\ref{cons:S}).
Similar to the analysis in Section \ref{Problem Formulation}, problem $\mathcal{P}2$ is NP-hard.

DRL stands as a classical method for addressing decision-making problems without knowledge of future information in dynamic environments \cite{9815021}. 
% Additionally, deep neural networks (DNNs) have demonstrated the capability to handle multi-dimensional inputs and extract intricate features [?]. 
Consequently, the joint optimization problem $\mathcal{P}2$ in MEC networks with dynamic tasks and services probably might be able to be effectively solved via a DRL-based method.
However, applying a DRL algorithm directly may not lead to a satisfactory solution for problem $\mathcal{P}2$ due to the following reason.
Because of the dynamic size of service data and limited storage capacity of ES, there are a number of valid combinations of downloading and caching decisions and the combinations change over time, which poses a significant challenge on DRL algorithms to devise the optimal combination that effectively accommodates the storage constraint \cite{huang2020closer}. 
% Also, to further show the impact of this issue and present the advantage of our proposed method, we conduct the compared experiments of directly applying DRL algorithm into our system in Section \ref{Sec:simulation}.

Therefore, we do not employ DRL algorithms to directly address downloading and caching decisions. Instead, we resort to the QCQP transformation and the SDR method to derive a rough idea of downloading and caching decisions.
With the QCQP transformation and SDR method, an effective solution with low-complexity can be provided as a rough idea.
% The SDR method can solve this transformed QCQP problem with low-complexity.
% and naturally provide the probability of downloading and caching.
% Based on the rough idea, we use a DRL algorithm to learn offloading decisions of user tasks which are constrained by caching decisions, and then determine the optimal decisions with the minimum cost.
Based on the rough idea, we use a DRL algorithm to obtain final decisions including offloading decision, downloading decision, caching decision, and computing resource allocation decision.
We refer to our proposed method as the online integrated optimization-DRL (OIODRL) for solving the temporally decoupled NP-hard problem $\mathcal{P}2$. This method comprises an optimization stage including the QCQP transformation and the SDR method, and a learning stage including the DRL algorithm.

% The framework of the OIODRL method is shown in Fig. \ref{fig:OIODRL}.

% Besides, compared to determining the computing resource allocation decision via the DRL algorithm, we use the convex optimization method, which not only reduce the action space of DRL algorithm but also provide the optimal solution of the computing resource allocation decision. 
% Specifically, given the downloading, caching, and offloading decisions from the DRL algorithm, we formulate a convex optimization to obtain the computing resource allocation decision and use the Lagrange multiplier method to solve it.

\subsection{Optimization Stage: QCQP Transformation and SDR}
In the optimization stage, we first simplify problem $\mathcal{P}2$, and then transform it into a QCQP problem and then relax the QCQP problem into a semidefinite programming (SDP) problem via the SDR method for providing a rough idea of downloading and caching decisions.
% We present in detail the optimization stage of the proposed OIODRL method.
% Firstly we transform problem $\mathcal{P}2$ into a QCQP problem and then relax the QCQP problem into an SDP problem for providing a rough idea of downloading and caching decisions.
% Before taking the QCQP transformation, we simply problem $\mathcal{P}2$.

For presentation simplicity, we omit the time symbol $t$ from notations such as $y_j(t)$, $z_j(t)$, and replace $z_j(t-1)$ with $z_j'$.
Besides, we define the notations below.
Let $\mathbf{e}_i$ denote an $I$-element unit column vector with the $i$-th
element being $1$ and other elements being $0$, and $\mathbf{G}_i$ denote an $I\times I$ matrix with the $i$-th diagonal element being $1$ and other elements being $0$.
Let $\mathbf{0}_{a\times b}$ denote an $a\times b$ matrix with all elements being $0$, and $\mathbf{1}_{a\times b}$ denote an $a\times b$ matrix with all elements being $1$.
% Let ${\rm Tr}(\mathbf{B})$ and ${\rm rank}(\mathbf{B})$ denote the trace function and rank function of matrix $\mathbf{B}$, respectively.
We use $\mathbf{b}(n)$ to denote the $n$-th element of vector $\mathbf{b}$ and $\mathbf{B}(n,m)$ to denote the $(n,m)$-th entry of matrix $\mathbf{B}$.
% Denote by $\mathbf{B}\succeq 0$ that $\mathbf{B}$ is a positive semi-definite matrix.

\subsubsection{Problem Simplification}
We have the following lemma.
% Specifically, we perform the following classification discussion, 
\begin{lemma}\label{lemma_cases}
In the optimal solution of problem $\mathcal{P}2$, we have the following three cases for any service $j\in\mathcal{J}$:
(\romannumeral1) when $z_j'=0$, then $y_j=z_j$;
(\romannumeral2) when $z_j'=1$ and $V\lambda_p p_j-H_j<0$, then $y_j=z_j$;
(\romannumeral3) when $z_j'=1$ and $V\lambda_p p_j-H_j\geq 0$, then $y_j=0$.
\end{lemma}
\begin{proof}
According to (\ref{cons:zy1}), we have case (\romannumeral1).
If $V\lambda_p p_j-H_j<0$, then downloading or refreshing service $j$ brings a positive gain to the objective function of problem $\mathcal{P}2$. 
Besides, due to constraint (\ref{cons:zy2}), service $j$ can be downloaded or refreshed (i.e., $y_j=1$) when and only when service $j$ is cached (i.e., $z_j=1$).
Thus, if service $j$ has been cached at the previous time slot (i.e., $z_j'=1$), then service $j$ should be removed (i.e., $y_j=0,z_j=0$) or should continuous to be cached while refreshed (i.e., $y_j=1,z_j=1$) to minimize the objective function in problem $\mathcal{P}2$, which yields case (\romannumeral2). 
Similar to case (\romannumeral2), if $V\lambda_p p_j-H_j\geq 0$, downloading or refreshing service $j$ would not bring a positive gain.
Thus, if service $j$ has been cached at the previous time slot (i.e., $z_j'=1$), then service $j$ should be removed (i.e., $y_j=0, z_j=0$) or should continuous to be cached while not refreshed (i.e., $y_j=0, z_j=1$), leading to case (\romannumeral3).
\end{proof}
% Therefore, we can eliminate downloading decision $y_j$ in problem $\mathcal{P}2$ according to different situations.
Here, we consider the term $g_1(\boldsymbol z(t))+g_2(\boldsymbol y(t))$ in the objective function of problem $\mathcal{P}2$.
Specifically, if case (\romannumeral1) or (\romannumeral2) in Lemma \ref{lemma_cases} is satisfied, we have 
$g_1(\boldsymbol z(t))+g_2(\boldsymbol y(t))=\sum_{j\in\mathcal{J}}V\lambda_p p_j z_j$; if case (\romannumeral3) in Lemma \ref{lemma_cases} is satisfied, we have 
$g_1(\boldsymbol z(t))+g_2(\boldsymbol y(t))=\sum_{j\in\mathcal{J}}H_j z_j$.
Hence, $g_1(\boldsymbol z(t))+g_2(\boldsymbol y(t))$ can be replaced by $\sum_{j\in\mathcal{J}}\mathcal{G}_j z_j$, in which $\mathcal{G}_j$ is defined as 
% variable parameter of $z_j$ as $\mathcal{G}_j$ for different situations,
\begin{equation}
\mathcal{G}_j=\nonumber
\begin{cases}
V\lambda_p p_j,&\text{if case (\romannumeral1) or (\romannumeral2) in Lemma \ref{lemma_cases} is satisfied},\\
H_j,&\text{if case (\romannumeral3) in Lemma \ref{lemma_cases} is satisfied}.
\end{cases}
\end{equation}
With this replacing, we can eliminate $y_j$ in problem $\mathcal{P}2$.

\subsubsection{QCQP Transformation}
To transform problem $\mathcal{P}2$ into a QCQP problem, we introduce an auxiliary variable denoted as $\bar D^\text{e}_{ij}$, and let $\boldsymbol{\bar D^\text{e}}(t)\triangleq[\bar D^\text{e}_{ij}]_{i\in\mathcal{I},j\in\mathcal{J}}$ and 
$\Lambda_{ij}\triangleq\lambda_D D_{ij}^\text{ec}+\lambda_DD_{ij}^\text{c}+\lambda_c\lambda_s S^\text{u}_{ij}$ for brevity. Then, problem $\mathcal{P}2$ is equivalent to problem $\mathcal{P}2.1$ as 
\begin{subequations}
\begin{align}
\mathcal{P}2.1:~
&\min_{\substack{\boldsymbol x(t),\boldsymbol z(t),\\ \boldsymbol f(t),\boldsymbol{\bar D^\text{e}}(t)}}~\sum_{j\in \mathcal{J}}[\mathcal{G}_j z_j-V(\sum_{i\in \mathcal{I}}\Lambda_{ij}x_{ij}-\lambda_D\bar D^\text{e}_{ij})]\nonumber\\
\text{s.t.}~~
&(\ref{cons:S}),(\ref{cons:xy}),(\ref{cons:F}),(\ref{cons:fij}),\notag\\
&x_{ij},z_{j} \in \{0,1\},~\forall i,j,\label{cons:xy_bi}\\
&\frac{Y_{ij}}{f^\text{e}_{ij}}x_{ij}\leq\bar D^\text{e}_{ij},~\forall i,j.\label{cons:aux}
\end{align}
\end{subequations}
% where $\boldsymbol{\bar D^\text{e}}(t)\triangleq[\bar D^\text{e}_{ij}]_{i\in\mathcal{I},j\in\mathcal{J}}$.
% Note that, the downloading decision $y_j$ and its constraints (\ref{cons:zy1}), (\ref{cons:zy2}) are removed from problem $\mathcal{P}2.1$ since $\mathcal{G}_j$ contains all situations of $y_j(t)$.
% Moreover, the constraint (\ref{cons:xyz}) is modified to constraint (\ref{cons:xy_bi}).
% Let $\mathbf{e}_i$ denote an $I$-element unit column vector with the $i$-th
% element being $1$ and other elements being $0$, and $\mathbf{G}_i$ denote an $I\times I$ matrix with the $i$-th diagonal element being $1$ and other elements being $0$.
% $\mathbf{0}_{a\times b}$ denote an $a\times b$ matrix with all elements being $0$, and $\mathbf{1}_{a\times b}$ denote an $a\times b$ matrix with all elements being $1$.
% % Let ${\rm Tr}(\mathbf{B})$ and ${\rm rank}(\mathbf{B})$ denote the trace function and rank function of matrix $\mathbf{B}$, respectively.
% We use $\mathbf{b}(n)$ to denote the $n$-th element of vector $\mathbf{b}$ and $\mathbf{B}(n,m)$ to denote $(n,m)$-th entry of matrix $\mathbf{B}$.
% % Denote by $\mathbf{B}\succeq 0$ that $\mathbf{B}$ is a positive semi-definite matrix.

Next, we rewrite the objective function of problem $\mathcal{P}2.1$ in the following vector form:
\begin{equation}\nonumber
\sum_{j\in \mathcal{J}}(\mathbf{b}_{j}^P)^T \mathbf{v}_{j},
\end{equation}
where $\mathbf{v}_{j}\triangleq [z_{j},x_{1j},...,x_{Ij},f_{1j},...,f_{Ij},\bar D^\text{e}_{1j},...,\bar D^\text{e}_{Ij}]^T,~\mathbf{b}^P_j\triangleq[\mathcal{G}_j,-V\Lambda_{1j},..., -V\Lambda_{Ij},\mathbf{0}_{1\times I},V\lambda_D\cdot\mathbf{1}_{1\times I}]^T$, and superscript $(\cdot)^T$ means transpose operation.
% where superscript $(\cdot)^T$ means transpose, and vectors $\mathbf{v}_{j}$ and $\mathbf{b}^P_j$ are defined as
% \begin{align}
% &\mathbf{v}_{j}\triangleq [z_{j},x_{1j},...,x_{Ij},f_{1j},...,f_{Ij},\bar D^\text{e}_{1j},...,\bar D^\text{e}_{Ij}]^T,\nonumber\\
% &\mathbf{b}^P_j\triangleq[\mathcal{G}_j,-V\Lambda_{1j},..., -V\Lambda_{Ij},\mathbf{0}_{1\times I},V\lambda_D\cdot\mathbf{1}_{1\times I}]^T.\nonumber
% \end{align}

In problem $\mathcal{P}2.1$, constraint (\ref{cons:xy}) can be rewritten as
\begin{equation}\nonumber
(\mathbf{b}^{xz}_{i})^T \mathbf{v}_{j}\leq 0,~\forall i,j,
\end{equation}
where $\mathbf{b}^{xz}_{i}\triangleq [-1,\mathbf{e}_i^T,\mathbf{0}_{1\times 2I}]^T$,
% \begin{equation}
% \mathbf{b}^{xz}_{i}\triangleq [-1,\mathbf{e}_i^T,\mathbf{0}_{1\times 2I}]^T,\nonumber
% \end{equation}
and constraints (\ref{cons:S}) and (\ref{cons:F}) can be rewritten as
\begin{equation}\nonumber
\sum_{j\in \mathcal{J}}(\mathbf{b}^s_j)^T \mathbf{v}_{j}\leq S,~
\sum_{j\in \mathcal{J}}(\mathbf{b}^f_j)^T \mathbf{v}_{j}\leq F,
\end{equation}
where $\mathbf{b}^s_j\triangleq [S_j,\mathbf{0}_{3\times I}]^T,~\mathbf{b}^f_j\triangleq [\mathbf{0}_{1+1\times I},\mathbf{1}_{1\times I},\mathbf{0}_{1\times I}]^T$.
% \begin{align}
% &\mathbf{b}^s_j\triangleq [S_j,\mathbf{0}_{3\times I}]^T,
% ~\mathbf{b}^f_j\triangleq [\mathbf{0}_{1+1\times I},\mathbf{1}_{1\times I},\mathbf{0}_{1\times I}]^T.\nonumber
% \end{align}
As for constraint (\ref{cons:xy_bi}), we rewrite it as $x_{ij}(1-x_{ij})=0,~z_{j}(1-z_{j})=0,~\forall i,j$, which is equivalent to 
% \begin{equation}
% x_{ij}(1-x_{ij})=0,~
% z_{j}(1-z_{j})=0,~\forall i,j,\nonumber
% \end{equation}
\begin{equation}
(\mathbf{b}_{i}^x)^T\mathbf{v}_{j}-\mathbf{v}_{j}^T\mathbf{B}_{i}^x\mathbf{v}_{j}=0,~
(\mathbf{b}^z_{i})^T\mathbf{v}_{j}-\mathbf{v}_{j}^T\mathbf{B}^z_{i}\mathbf{v}_{j}=0,~\forall i,j,\nonumber
\end{equation}
where $\mathbf{b}^x_i\triangleq [0,\mathbf{e}_{i}^T,\mathbf{0}_{2\times I}]^T,~\mathbf{b}^z_{i}\triangleq [1,\mathbf{0}_{3\times I}]^T$,
\begin{align}
% &\mathbf{b}^x_i\triangleq [0,\mathbf{e}_{i}^T,\mathbf{0}_{2\times I}]^T,
% ~\mathbf{b}^z_{i}\triangleq [1,\mathbf{0}_{3\times I}]^T,\nonumber\\
&\mathbf{B}^{x}_{i}\triangleq
\begin{bmatrix}\nonumber
0&\mathbf{0}_{1\times I}&\mathbf{0}_{1\times 2I}\\
\mathbf{0}_{I\times 1}&\mathbf{G}_i&\mathbf{0}_{I\times 2I}\\
\mathbf{0}_{2I\times 1}&\mathbf{0}_{2I\times I}&\mathbf{0}_{2I\times 2I}
\end{bmatrix},
~\mathbf{B}^{z}_{i}\triangleq
\begin{bmatrix}\nonumber
1&\mathbf{0}_{1\times 3I}\\
\mathbf{0}_{3I\times 1}&\mathbf{0}_{3I\times 3I}
\end{bmatrix}.
\end{align}
The auxiliary variable constraint (\ref{cons:aux}) is rewritten as
\begin{equation}\nonumber
(\mathbf{b}_{ij}^\text{e})^T\mathbf{v}_{j}-\mathbf{v}_{j}^T\mathbf{B}_{i}^\text{e}\mathbf{v}_{j}\leq 0,~\forall i,j,
\end{equation}
where $\mathbf{b}^\text{e}_{ij}\triangleq [0,Y_{ij}\mathbf{e}_i^T,\mathbf{0}_{2\times I}]^T$,
\begin{align}
% &\mathbf{b}^\text{e}_{ij}\triangleq [0,Y_{ij}\mathbf{e}_i^T,\mathbf{0}_{2\times I}]^T,\nonumber\\
&\mathbf{B}^{e}_{i}\triangleq
\begin{bmatrix}\nonumber
\mathbf{0}_{(I+1)\times (I+1)}&\mathbf{0}_{(I+1)\times I}&\mathbf{0}_{(I+1)\times I}\\
\mathbf{0}_{I\times (I+1)}&\mathbf{0}_{I\times I}&\frac{1}{2}\mathbf{G}_i\\
\mathbf{0}_{I\times (I+1)}&\frac{1}{2}\mathbf{G}_i&\mathbf{0}_{I\times I}
\end{bmatrix}.
\end{align}
By defining vector $\mathbf{u}_j\triangleq[\mathbf{v}_j^T,1]^T$, for $j\in\mathcal{J}$, problem $\mathcal{P}2.1$ is transformed into an equivalent QCQP problem $\mathcal{P}2.2$ as
\begin{subequations}\label{QCQP}
\begin{align}
\mathcal{P}2.2:~
&\min\limits_{\{\mathbf{u}_{j}\}}~\sum_{j\in \mathcal{J}}\mathbf{u}_{j}^T \mathbf{F}_{j}^P \mathbf{u}_{j}\nonumber\\
{\rm s.t.}~~
&\mathbf{u}_j^T\mathbf{F}^{xz}_j \mathbf{u}_j\leq 0,~\forall j,\\
&\sum_{j\in \mathcal{J}} \mathbf{u}_j^T \mathbf{F}^s_j \mathbf{u}_j\leq S,\\
&\sum_{j\in \mathcal{J}} \mathbf{u}_j^T \mathbf{F}^f_j \mathbf{u}_j\leq F,\\
&\mathbf{u}_j^T\mathbf{F}^\sigma_i \mathbf{u}_j=0,~\sigma\in\{x,z\},~\forall i,j,\label{cons:QCQP_binary}\\
&\mathbf{u}_j^T \mathbf{F}^\text{e}_{ij} \mathbf{u}_j\leq 0,~\forall i,j,\\
&\mathbf{u}_{j}\succeq 0,~\forall j,
\end{align}
\end{subequations}
where $\mathbf{F}_j^\rho\triangleq \nonumber
\begin{bmatrix}
\mathbf{0}_{(3I+1)\times (3I+1)}&\frac{1}{2}\mathbf{b}^\rho_j\\
\frac{1}{2}(\mathbf{b}^\rho_j)^T&0
\end{bmatrix},~\rho\in\{P,s,f\},$
\begin{align}
% &\mathbf{F}_j^\rho\triangleq \nonumber
% \begin{bmatrix}
% \mathbf{0}_{(3I+1)\times (3I+1)}&\frac{1}{2}\mathbf{b}^\rho_j\\
% \frac{1}{2}(\mathbf{b}^\rho_j)^T&0
% \end{bmatrix},~\rho\in\{P,s,f\},\\
&\mathbf{F}_i^{xz}\triangleq \nonumber
\begin{bmatrix}
\mathbf{0}_{(3I+1)\times (3I+1)}&\frac{1}{2}\mathbf{b}^{xz}_i\\
\frac{1}{2}(\mathbf{b}^{xz}_i)^T&0
\end{bmatrix},
~\mathbf{F}_{ij}^\text{e}\triangleq \nonumber
\frac{1}{2}\begin{bmatrix}
-2\mathbf{B}_i^\text{e}&\mathbf{b}_{ij}^\text{e}\\
(\mathbf{b}_{ij}^\text{e})^T&0
\end{bmatrix},\\
&\mathbf{F}_i^\sigma\triangleq \nonumber
\frac{1}{2}\begin{bmatrix}
-2\mathbf{B}_i^\sigma&\mathbf{b}_i^\sigma\\
(\mathbf{b}_i^\sigma)^T&0
\end{bmatrix},~\sigma\in\{x,z\}.
\end{align}

\subsubsection{Semidefinite Relaxation}
Problem $\mathcal{P}2.2$ is still NP-hard.
To get an approximate solution, we adopt the SDR method.
In problem $\mathcal{P}2.2$, we have several notations in the format of 
$\mathbf{u}_{j}^T\mathbf{F}\mathbf{u}_{j}$ with $\mathbf{F}$ denote a $(3I+2)\times(3I+2)$ matrix.
By defining 
$\mathbf{U}_j\triangleq\mathbf{u}_j \mathbf{u}_j^T$, $\mathbf{u}_{j}^T\mathbf{F}\mathbf{u}_{j}$ can be rewritten as
\begin{equation}
\mathbf{u}_{j}^T\mathbf{F}\mathbf{u}_{j}={\rm Tr}(\mathbf{F}\mathbf{U}_{j}),\nonumber
\end{equation}
where ${\rm Tr}(\cdot)$ means trace operation and ${\rm rank}(\mathbf{U}_{j})=1$. Ignoring rank constraint ${\rm rank}(\mathbf{U}_{j})=1$, we relax problem $\mathcal{P}2.2$ to an SDP problem $\mathcal{P}3$ as
\begin{subequations}\label{SDP}
\begin{align}
\mathcal{P}3:~
&\min\limits_{\{\mathbf{U}_{j}\}}~\sum_{j\in \mathcal{J}}{\rm Tr}(\mathbf{F}_{j}^P\mathbf{U}_{j})\nonumber\\
{\rm s.t.}~~
&{\rm Tr}(\mathbf{F}_{j}^{xz}\mathbf{U}_{j})\leq 0,~\forall j,\label{cons:SDP_xz}\\
&\sum_{j\in \mathcal{J}}{\rm Tr}(\mathbf{F}_{j}^s\mathbf{U}_{j})\leq S,\\
&\sum_{j\in \mathcal{J}}{\rm Tr}(\mathbf{F}_{j}^f\mathbf{U}_{j})\leq F,\\
&{\rm Tr}(\mathbf{F}_{i}^\sigma\mathbf{U}_{j})=0,~\sigma\in\{x,z\},~\forall i,j,\label{cons:SDP_binary}\\
&{\rm Tr}(\mathbf{F}_{ij}^\text{e}\mathbf{U}_{j})\leq 0,~\forall i,j,\label{cons:SDP_e}\\
% &{\rm Tr}(\mathbf{F}_{i}^\phi\mathbf{U}_{j})=0,~\phi\in\{F,D\},~\forall i,j,\label{cons:SDP_FD}\\
&{\rm Tr}(\mathbf{F}_{j}^z\mathbf{U}_{j})\leq 1,~\forall j,\label{cons:SDP_z}\\
&\mathbf{U}_{j}(3I+2,3I+2)=1,~\forall j,\\
&\mathbf{U}_{j}\succeq 0,~\forall j,\label{cons:SDP_semi}
\end{align}
\end{subequations}
where $\mathbf{F}_j^z\triangleq \nonumber
\frac{1}{2}\begin{bmatrix}
0&\mathbf{0}_{1\times 3I}&1\nonumber\\
\mathbf{0}_{3I\times 1}&\mathbf{0}_{3I\times 3I}&\mathbf{0}_{3I\times 1}\nonumber\\
1&\mathbf{0}_{1\times 3I}&0
\end{bmatrix}.$
% \begin{align}
% &\mathbf{b}^F_i\triangleq [\mathbf{0}_{1+1\times I},\mathbf{e}_{i}^T,\mathbf{0}_{1\times I}]^T,
% ~\mathbf{b}^D_i\triangleq [\mathbf{0}_{1+2\times I},\mathbf{e}_{i}^T]^T,\nonumber\\
% &\mathbf{B}^{F}_{i}\triangleq
% \begin{bmatrix}
% \mathbf{0}_{(1+I)\times (1+I)}&\mathbf{0}_{(1+I)\times I}&\mathbf{0}_{(1+I)\times I}\nonumber\\
% \mathbf{0}_{I\times (1+I)}&{\rm diag}(\mathbf{e}_i)&\mathbf{0}_{I\times I}\\
% \mathbf{0}_{I\times (1+I)}&\mathbf{0}_{I\times I}&\mathbf{0}_{I\times I}
% \end{bmatrix},\nonumber\\
% &\mathbf{B}^{D}_{i}\triangleq
% \begin{bmatrix}
% \mathbf{0}_{(1+2I)\times (1+2I)}&\mathbf{0}_{(1+2I)\times I}\nonumber\\
% \mathbf{0}_{I\times (1+2I)}&{\rm diag}(\mathbf{e}_i)
% \end{bmatrix},\nonumber\\
% &\mathbf{F}_i^\phi\triangleq \nonumber
% \frac{1}{2}\begin{bmatrix}
% -2\mathbf{B}_i^\phi&\mathbf{b}_i^\phi\nonumber\\
% (\mathbf{b}_i^\phi)^T&0
% \end{bmatrix},~\phi\in\{F,D\}.
% &\mathbf{F}_j^z\triangleq \nonumber
% \frac{1}{2}\begin{bmatrix}
% 0&\mathbf{0}_{1\times 3I}&1\nonumber\\
% \mathbf{0}_{3I\times 1}&\mathbf{0}_{3I\times 3I}&\mathbf{0}_{3I\times 1}\nonumber\\
% 1&\mathbf{0}_{1\times 3I}&0
% \end{bmatrix}.
% \end{align}

We can efficiently solve problem $\mathcal{P}3$ in polynomial time using SDP optimization toolboxes, e.g., YALMIP, SeDuMi.
Denote the optimal solution of problem $\mathcal{P}3$ as $\{\mathbf{U}_{j}^\text{sdp}\}$.
Then, from the solution $\{\mathbf{U}_{j}^\text{sdp}\}$ of problem $\mathcal{P}3$ and Lemma \ref{lemma_cases}, we can obtain a rough idea of downloading and caching decisions for problem $\mathcal{P}2.2$ (also for problems $\mathcal{P}2$, $\mathcal{P}2.1$) as follows.

Specifically, problem $\mathcal{P}3$ is a relaxed version of problem $\mathcal{P}2.2$ since the rank constraint is discarded from problem $\mathcal{P}3$. Thus, the solution provided in problem $\mathcal{P}3$ may be infeasible in problem $\mathcal{P}2.2$ (also in problems $\mathcal{P}2$, $\mathcal{P}2.1$). 
From the definition $\mathbf{U}_j\triangleq\mathbf{u}_j \mathbf{u}_j^T$ and constraint $\mathbf{U}_{j}(3I+2,3I+2)=1$, we can obtain that the last row of $\mathbf{U}_j$ equals $\mathbf{u}_j$, i.e., $\mathbf{U}_j(3I+2,n)=\mathbf{u}_j(n),~\forall j\in\mathcal{J},n\in\{1,...,3I+2\}$.
Based on this, let $\mathbf{u}_j^\text{sdp}(n)\triangleq\mathbf{U}^\text{sdp}_j(3I+2,n),~\forall j,n$, and we can recover caching decision $z_j$ from $\mathbf{u}^\text{sdp}_j(1)=\mathbf{U}^\text{sdp}_j(3I+2,1)$ and then obtain downloading decision $y_j$ according to Lemma \ref{lemma_cases}.
% Before recovering caching decision, we present the reason that the solution obtained by problem $\mathcal{P}3$ is probably infeasible in problem $\mathcal{P}2.2$ ($\mathcal{P}2$, $\mathcal{P}2.1$), shown as Lemma \ref{lemma_relax}. 

\begin{remark}\label{lemma_relax}
In problem $\mathcal{P}2.2$ (and problems $\mathcal{P}2$, $\mathcal{P}2.1$), offloading decision $x_{ij}$ and caching decision $z_j$ are binary (i.e., $x_{ij},z_{j}\in\{0, 1\},~\forall i,j$). However, based on the solution $\{\mathbf{U}^\text{sdp}_j\}$ of problem $\mathcal{P}3$, the obtained solutions $x_{ij}$ and $z_j$ are continuous values with range $[0,1],~\forall i,j$. The reason is given below.
% Compared to problem $\mathcal{P}2.2$ ($\mathcal{P}2$, $\mathcal{P}2.1$), the binary constraint (\ref{cons:xy_bi}) of offloading decision $x_{ij}$ and caching decision $z_j$ is relaxed to $x_{ij},z_{j} \in [0,1],~\forall i,j$ in problem $\mathcal{P}3$.
% and the delay cost of the edge process $\frac{Y_{ij}}{f^\text{e}_{ij}}x_{ij}$ is relaxed to an auxiliary variable $\bar D^\text{e}_{ij}$ and is bounded by a variable related to $\bar D^\text{e}_{ij}$.
\end{remark}
% Denote the optimal solution of problem $\mathcal{P}3$ and the corresponding relaxed solution for problem $\mathcal{P}2.2$ by $\{\mathbf{U}_{j}^{sdp}\}$ and $\{\mathbf{u}_{j}^{sdp}\}$, respectively.
From semidefinite constraint (\ref{cons:SDP_semi}) in problem $\mathcal{P}3$, matrix $\mathbf{U}_{j}^\text{sdp}$ is symmetric and $\mathbf{U}_j^\text{sdp}(n,n)\geq 0,~\forall j,n$.
Then, we have $\mathbf{u}_j^\text{sdp}(n)=\mathbf{U}_j^\text{sdp}(3I+2,n)=\mathbf{U}_j^\text{sdp}(n,3I+2),~\forall j,n$.
Due to the discarded rank constraint, constraint $\mathbf{u}_j(n)^2=\mathbf{U}_j(n,n)$ is removed, thus constraint (\ref{cons:SDP_binary}) only ensures that $\mathbf{u}_j^\text{sdp}(n)=\mathbf{U}_j^\text{sdp}(n,n)\geq 0,~\forall j,n$ rather than $\mathbf{u}_j^\text{sdp}(n)=\mathbf{u}_j^\text{sdp}(n)^2,~\forall j,n$, i.e., $x_{ij}$ and $z_{j}$ are non-negative but not necessarily binary.
Together with constraints (\ref{cons:SDP_xz}) and (\ref{cons:SDP_z}), we have $x_{ij}\leq z_j$ and $z_j\leq 1$, $\forall i,j$, and thus $x_{ij}, z_j\in[0,1], \forall i,j$.
% which leads to the conclusion of Lemma \ref{lemma_relax}.
% For the same reason, the constraint $\mathbf{u}_j(n)\mathbf{u}_j(m)=\mathbf{U}_j(n,m)=\mathbf{U}_j(m,n)$ is removed.
% Therefore, the constraint (\ref{cons:SDP_e}) only guarantees $Y_{nj}\mathbf{u}_j^{sdp}(n)\leq\mathbf{U}_j^{sdp}(I+n,2I+n),n\in[2,1+I]$ but not $Y_{nj}\mathbf{u}_j^{sdp}(n)\leq\mathbf{u}_j^{sdp}(I+n)\mathbf{u}_j^{sdp}(2I+n),n\in[2,1+I]$, which means that $\frac{Y_{ij}}{f^\text{e}_{ij}}x_{ij}$ is relaxed to $\bar D^\text{e}_{ij}$ while not being constrained directly by it, but is bounded by a $\bar D^\text{e}_{ij}$-related variable due to the semidefinite constraint (\ref{cons:SDP_semi}).
% Notably, in problem $\mathcal{P}2.2$, the $\bar D^\text{e}_{ij}$-related variable is equal to $\bar D^\text{e}_{ij}$, which is a special case of $\mathcal{P}3$.
% Besides, the constraint (\ref{cons:SDP_FD}) is introduced to ensure that $f_{ij}$ and $\bar D^\text{e}_{ij}$ are non-negative.

Recall that we use the optimization stage to obtain a rough idea of caching and downloading decisions $z_j,y_j,~\forall j$. 
Since $y_j$ can be obtained from $z_j$ according to Lemma \ref{lemma_cases}, we only focus on $z_j$ here.
As the obtained solutions $z_j$'s are continuous values, we treat the obtained solution $z_j$ as the probability to cache service $j$.
According to the obtained $z_1,z_2,...,z_J$ (the probabilities of caching the $J$ services), we take $K$ samples of caching decisions. In each sample, we have $J$ binary caching decisions for the $J$ services, and further check the feasibility, i.e., whether or not the caching decisions violate the storage capacity constraint. If the sample is not feasible, we randomly select a cached service of this sample (e.g., service $j$ with $z_j=1$ in the sample) and remove it (i.e., set $z_j=0$ in the sample). We repeat the feasibility checking and remove until eventually the storage capacity constraint is satisfied. Then we get the corresponding downloading decisions for the sample.

As a summary, at the end of the optimization stage, we get $K$ samples, and in each sample, we have downloading and caching decisions for each service, which are denoted as $\widetilde{y}_j(t)$ and $\widetilde{z}_j(t)$. 
Besides, denote the $K$ samples of $\widetilde{y}_j(t)$ and $\widetilde{z}_j(t)$ for all services as $\boldsymbol{\widetilde{y}}(t)$ and $\boldsymbol{\widetilde{z}}(t)$.
$\boldsymbol{\widetilde{y}}(t)$ and $\boldsymbol{\widetilde{z}}(t)$ are our rough idea about downloading and caching decisions and are input to the next stage, i.e. the learning stage.

\subsection{Learning Stage: Clip-PPO2}
With the rough idea from the optimization stage, the system evolution over time slots can be modeled as a Markov decision process (MDP), which has four basic elements as follows:
% The four basic elements in the MDP, i.e., \textit{agent}, \textit{state space}, \textit{action space}, and \textit{reward function}, are given as follows: 
\begin{itemize}
    \item\textit{Agent}: 
    In the considered MEC system, the ES is regarded as an agent. The agent continues interacting with the environment of system. The agent gets an observation of the  system state over each time slot and takes an action according to its policy. The system environment provides the agent with a reward after the action taking.
    Then the system state transits to next state at next time slot. The agent tries to learn the optimal policy for obtaining more rewards. 
    \item\textit{State space}: 
    The state describes the status of the stochastic and dynamic MEC system.
    % In the learning stage of OIODRL, the agent obtains the state from the current system environment. 
    Accordingly, we define the state at time slot $t$ as $s_t\triangleq[\boldsymbol{\xi}(t),\boldsymbol{\zeta}(t),\boldsymbol{z}(t-1),\boldsymbol{O}(t), \boldsymbol{\widetilde{y}}(t), \boldsymbol{\widetilde{z}}(t)]$, in which $\boldsymbol{\xi}(t)\triangleq[S^\text{u}_{ij}(t),S^\text{d}_{ij}(t),Y_{ij}(t)]_{i\in\mathcal{I},j\in\mathcal{J}}$ denotes the task information, $\boldsymbol{\zeta}(t)\triangleq[S_{j}(t),p_j^\text{p}(t),p_j^\text{r}(t)]_{j\in\mathcal{J}}$ denotes the service attribution, $\boldsymbol{z}(t-1)$ denotes the service caching status (i.e., the caching decisions of time slot $t-1$), $\boldsymbol{O}(t)$ denotes the AoI observation mentioned in Section \ref{Sec:Lyapunov}, and $\boldsymbol{\widetilde{y}}(t)$ and $\boldsymbol{\widetilde{z}}(t)$ are given by the optimization stage.
    
    % \begin{itemize}
    %     \item In the optimization stage, the state at time slot $t$ is defined as $s_t^\text{o}\triangleq[\boldsymbol{\xi}(t),\boldsymbol{\zeta}(t),\boldsymbol{z}(t-1),\boldsymbol{O}(t)]]$, in which $\boldsymbol{\xi}(t)\triangleq[S^\text{u}_{ij}(t),S^\text{d}_{ij}(t),Y_{ij}(t)]_{i\in\mathcal{I},j\in\mathcal{J}}$ denotes the task information, $\boldsymbol{\zeta}(t)\triangleq[S_{j}(t),p_j^\text{p}(t),p_j^\text{r}(t)]_{j\in\mathcal{J}}$ denotes the service attribution, $\boldsymbol{z}(t-1)$ denotes the service caching status, and $\boldsymbol{O}(t)$ denotes the AoI observation mentioned in Section \ref{Sec:Lyapunov}.
    %     \item In the learning stage, the state at time slot $t$ is defined as $s_t^\text{l}\triangleq[\boldsymbol{\xi}(t),a_t^\text{o}]$, in which $a_t^\text{o}$ denotes the action provided in the optimization stage and will be defined in the sequel.
    % \end{itemize}

    \item\textit{Action space}:
    The action describes the agent's (ES's) decisions.
    % In the learning stage of OIODRL, the agent observes states $s_t$ and then outputs the action.
    The action at time slot $t$ is defined as $a_t\triangleq[\boldsymbol{x}(t),\boldsymbol{y}(t),\boldsymbol{z}(t),\boldsymbol{f}(t)]$, which includes offloading decision $\boldsymbol{x}(t)$, downloading decision $\boldsymbol{y}(t)$, caching decision $\boldsymbol{z}(t)$, and computing rate allocation decision $\boldsymbol{f}(t)$. 
    $a_t$ is the final decision that will be taken by the agent.
    
    % \begin{itemize}
    %     \item In the optimization stage, the action at time slot $t$ is defined as $a_t^\text{o}\triangleq[\Tilde{\boldsymbol{y}}(t),\Tilde{\boldsymbol{z}}(t)]$, in which $\Tilde{\boldsymbol{y}}(t)$ and $\Tilde{\boldsymbol{z}}(t)$ are the rough idea of downloading and caching decisions, and are not the final decisions of the agent.
    %     \item In the learning stage, the action at time slot $t$ is defined as $a_t^\text{l}\triangleq[\boldsymbol{x}(t),\boldsymbol{y}(t),\boldsymbol{z}(t),\boldsymbol{f}(t)]$, which includes offloading, downloading, caching, and computing rate allocation decisions, and are the final decisions that will be taken by the agent.
    % \end{itemize}

    \item\textit{Reward function}:
    The reward function evaluates the system's actions.
    Since the objective function of problem $\mathcal{P}2$ is to minimize the cost, we intuitively set the negative of problem $\mathcal{P}2$'s objective function as the reward. Thus, the reward after taking action $a_t$ under state $s_t$ at time slot $t$ is defined as $r(s_t,a_t)\triangleq-[g_1(\boldsymbol{z}(t))+g_2(\boldsymbol{y}(t))-g_3(\boldsymbol{x}(t),\boldsymbol{f}(t))]$.
    % \footnote{The task information $\boldsymbol{\xi}(t)$ in state $s_t^\text{l}$ has been contained in state $s_t^\text{o}$ and the intermediate action $a_t^\text{o}$ in state $s_t^\text{l}$ is not needed to calculate the reward. Thus, the reward function only takes the state $s_t^\text{o}$ and action $a_t^\text{l}$ as the input.}
    % Specifically, in our proposed OIODRL method, the agent observes the state $s_t^\text{o}$ and takes the action $a_t^\text{l}$, and then obtains the reward $r(s_t^\text{o},a_t^\text{l})$.
    Recall that, $\boldsymbol{\widetilde{y}}(t)$ and $\boldsymbol{\widetilde{z}}(t)$ in state $s_t$ are the rough idea, while $\boldsymbol{y}(t)$ and $\boldsymbol{z}(t)$ in action $a_t$ are the final decisions. So the reward function is based on $\boldsymbol{y}(t)$ and $\boldsymbol{z}(t)$ rather than $\boldsymbol{\widetilde{y}}(t)$ and $\boldsymbol{\widetilde{z}}(t)$.
    % In practice, the reward function does not take $\boldsymbol{\widetilde{y}}(t)$ and $\boldsymbol{\widetilde{z}}(t)$ in the state $s_t$ since they are not the final decisions but used for providing insight to the agent in order to obtain a wise action.
\end{itemize}

% Before presenting the detailed learning stage of OIODRL method, we first reformulate the problem $\mathcal{P}2$ as a Markov decision process (MDP).
% The state of the dynamic MEC system is effected by various task and service, and varying decisions of ES.
% Specifically, the new stochastic state is triggered jointly by the previous state and previous action of the system, and the action only depends on current system state \cite{9895362}. 
% Then, the stochastic and dynamic MEC system follows the Markov property and the problem $\mathcal{P}2$ can be modeled as an MDP.

To solve problem $\mathcal{P}2$, we will solve the MDP problem by maximizing the reward.
To this end, DRL algorithms are powerful tools, such as deep Q-network (DQN), REINFORCE, advantage actor-critic (A2C), and proximal policy optimization (PPO).
% There are a series of DRL algorithms, such as off-policy and valued-based algorithm deep Q-network (DQN), while it is not suitable for our problem. 
DQN is not suitable for our problem since it is an off-policy algorithm which is prone to overestimation bias \cite{anschel2017averaged} especially in our stochastic and dynamic MEC system. Thus, DQN can lead to suboptimal policies and slow convergence.
% DQN is an off-policy and valued-based algorithm, while it is not suitable for our problem. 
% It learns from data collected under a different policy rather than the one currently being optimized.
% Thus, DQN is prone to overestimation bias especially in our stochastic and dynamic MEC system, which can lead to suboptimal policies and slow convergence.
REINFORCE, which is an on-policy and policy-based algorithm, may also be inappropriate to solve our problem, since it may suffer from high variance \cite{tucker2017rebar} when learning complex policies in our complex MEC system.
% on-policy and policy-based algorithms like REINFORCE may also not be appropriate to solve our problem since they may suffer from high variance when learning complex policies in our complex MEC system.
A2C is a combined value-based and policy-based algorithm and may be unstable in the training process especially in our MEC system with high-dimensional state space and complex dynamics.
% combined valued-based and policy-based algorithms such as advantage actor-critic (A2C) may be unstable in the training process especially in our MEC system with high-dimensional state spaces and complex dynamics.
% On the other hand, PPO algorithm has the following benefits.
% Thus, we resort to the clip-PPO2 algorithm that is based on an actor-critic style for the following reasons.
On the other hand, PPO is more reliable than policy-based algorithms, and when using clipped objective functions, it can enhance stability, prevent significant performance degradation, and is easy for implementation \cite{schulman2017proximal,10322784}.
% takes the advantages of actor-critic architecture and clipped objective functions.
% The actor-critic architecture contributes PPO to be more reliable than policy-based algorithms \cite{?}, and clipped objective functions enhance stability, prevent significant performance degradation, and is easy for implementation \cite{?}.
% By using an actor network to improve policies and a critic network to estimate state values, the critic network contributes to guiding the actor's learning, and thus the clip-PPO2 algorithm is more reliable than policy-based algorithms.
% In addition, the clip-PPO2 algorithm clips the loss function to mitigate excessively large policy updates. This method enhances stability and prevents significant performance degradation, and is easy for implementation.
% The above strengths make the PPO highly effective in addressing complex problems and it has demonstrated state-of-the-art performance across various domains, including robotics, autonomous driving, network optimization, and beyond.
Therefore, we propose to use the clip-PPO2 algorithm (a kind of PPO algorithms) to solve the MDP problem in the learning stage of OIODRL.
For simplicity, PPO refers to clip-PPO2 in the sequel.

% Then, based on the rough idea of downloading and caching decisions given by the optimization stage, we adopt the clip-PPO2 algorithm in the learning stage of OIODRL method to learn the offloading decision and then determine the offloading, downloading, caching, and computing resource allocation decisions.

The adopted PPO algorithm is based on an actor-critic architecture \cite{schulman2017proximal}, containing three modules: \textit{actor network} that outputs a policy, \textit{critic network} that evaluates the policy, and \textit{updating network} that improves the actor and critic networks.
% The training process is shown as below.
% The agent interacts with the environment continuously using the policy given by the actor network and wins the reward from the environment, and then the environment evolves to the next state.
% During the interaction of each time slot, the corresponding action, state, reward, and next state are stored as a \textit{transition} in the \textit{replay buffer}.
% After every $N$ episodes, we sample a batch of transitions from the replay buffer to update the networks.
% Specifically, given the transitions, the critic network performs the estimation of the \textit{state value} and calculate \textit{advantage} that defined later, and the actor network performs the outputting of the updated policy.
% The state value is the expected reward of the system evolve from the state
% Based on this, the loss function is constructed, and then the loss and its gradient can be obtained to update the actor and critic networks via the stochastic gradient update technique.

\subsubsection{Actor network}
The actor network is utilized to output the policy for the system actions.
The actor network comprises of a deep neural network (DNN), an action mask, and an optimization module, as follows. 
% \footnote{In the vanilla PPO algorithm, a shared-network architecture is set to actor and critic networks by default. While compared to a separated-network architecture, the shared-network architecture may perform worse since there is objective competition between the actor and critic networks. Thus, we adopt a separated-network architecture in the proposed method.}
\begin{itemize}
    \item DNN: At each time slot $t$, the DNN receives the input, i.e., state $s_t$.\footnote{From $s_t$, the DNN only takes the task information $\boldsymbol{\xi}(t)$ and the rough idea of downloading and caching decisions $\boldsymbol{\widetilde{y}}(t),\boldsymbol{\widetilde{z}}(t)$, but does not take $\boldsymbol{\zeta}(t)$, $\boldsymbol{z}(t-1)$, $\boldsymbol{O}(t)$, since $\boldsymbol{\widetilde{y}}(t)$ and $\boldsymbol{\widetilde{z}}(t)$ already contain the information of $\boldsymbol{\zeta}(t)$, $\boldsymbol{z}(t-1)$, $\boldsymbol{O}(t)$.}
    Note that, in $s_t$, $\boldsymbol{\widetilde{y}}(t)$ and $\boldsymbol{\widetilde{z}}(t)$ includes the $K$ groups of rough ideas of downloading and caching decisions for all services. 
    Then the DNN performs forward propagation to output $K$ groups of the probability of offloading decisions. 
    To output the probability of offloading decisions, the sigmoid activation function is utilized in the output layer of DNN.
    Then the $K$ groups of the probability of offloading decisions are mapped to $K$ groups of binary offloading decisions by sampling the offloading decisions with the probability values.
    Thus, the DNN of the actor network can be expressed as $\pi_{\theta}:s_t\mapsto x_{ij}(t),\forall i, j$, where $\theta$ means the parameters of DNN. The DNN outputs totally $I$ binary offloading decisions for the $I$ users, since at each time slot, each user only requests one service type.  
    % In addition, let $\hat{\theta}$ denote the parameters of DNN in the last training round and $\pi_{\theta}(x_{ij}|s_t)$ denote the probability of outputting offloading decision $x_{ij}$ given state $s_t$ under policy $\pi_{\theta}$.
    % In practice, the DNN only outputs actions of $I$ dimensions rather than $2J+2IJ$ dimensions. 
    % Since each user can request at most one service type per time slot, there are $I$ tasks per time slot, and then the offloading decisions for the $I$ tasks (the output actions of $I$ dimensions) can be mapped into $IJ$ dimensions according to $J$ services.
    % Besides, the other $2J+IJ$ actions ($y_j(t)$, $z_j(t)$, $f^\text{e}_{ij}(t)$) are given by the optimization module, as shown later.
    % In this way, the action space is reduced significantly. 

    \item Action mask: 
    % However, the offloading decision are subject to the caching decision (the input rough idea), i.e., constraint (\ref{cons:xy}) should be satisfied.
    However, the offloading decision should satisfy constraint (\ref{cons:xy}), which means that the ES can process a task only when it has cached the associated service data.
    To guarantee constraint (\ref{cons:xy}) is always satisfied, the action mask is applied to the offloading decision as follows:
    % If constraint (\ref{cons:xy}) is not satisfied, it means that a service (say service $j$) is not cached at the ES while the associated task (say user $i$'s task $j$) is decided to be processed at the ES. % which is called an invalid offloading decision.
    % To guarantee constraint (\ref{cons:xy}) is always satisfied, the action mask is applied to the offloading decision as follows:
    % An invalid offloading decision refers to a service (say service $j$) is not cached at the ES while the associated task (say user $i$'s task $j$) is decided to be processed at the ES. 
    % we utilize the following action mask to obtain a valid offloading decision, denoted by $x_{ij}^\text{mask}(t)$:
    \begin{align}
    x_{ij}(t) \leftarrow x_{ij}(t) \widetilde{z}_j(t),~\forall i,j. \nonumber
    \end{align}
    \item Optimization module: 
    The purpose of the optimization module is to obtain computing rate decisions $f^\text{e}_{ij}(t)$.
    Specifically, given the binary decisions $\widetilde{y}_{j}(t)$, $\widetilde{z}_{j}(t)$, $x_{ij}(t)$, we solve problem $\mathcal{P}2$ to get $f^\text{e}_{ij}(t)$ as follows.
    Since $g_1(\boldsymbol z(t))$ and $g_2(\boldsymbol y(t))$ in problem $\mathcal{P}2$ are independent of $f^\text{e}_{ij}(t)$ and only term $\frac{Y_{ij}(t)}{f^\text{e}_{ij}(t)}x_{ij}(t)$ in $g_3(\boldsymbol x(t), \boldsymbol f(t))$ is related to $f^\text{e}_{ij}(t)$, problem $\mathcal{P}2$ is equivalent to
    % Then, we have $K$ groups of downloading $\widetilde{y}_j(t)$, caching $\widetilde{z}_j(t)$, and offloading decisions $x_{ij}(t)$, yet still not determine computing resource allocation decisions $f^\text{e}_{ij}(t)$.
    % To this end, we use an optimization module.
    % Given the binary decisions $\widetilde{y}_{j}(t), \widetilde{z}_{j}(t), x_{ij}(t)$, problem $\mathcal{P}2$ can be solved by maximizing term $g_3(\boldsymbol x(t), \boldsymbol f(t))$ through selecting $f^\text{e}_{ij}(t)$.
    % Since the other terms $g_1(\boldsymbol z(t))$ and $g_2(\boldsymbol y(t))$ in problem $\mathcal{P}2$ are determined and only term $g_3(\boldsymbol x(t), \boldsymbol f(t))$ is related to $f^\text{e}_{ij}(t)$. 
    % Therefore, after giving binary decisions and ignoring $f^\text{e}_{ij}(t)$-independent variables, problem $\mathcal{P}2$ is equivalent to
    \begin{align}
    \min_{\boldsymbol f(t)}~&\sum_{i\in \mathcal{I}}\sum_{j\in\mathcal{J}}\frac{Y_{ij}(t)}{f^\text{e}_{ij}(t)}x_{ij}(t)\nonumber\\
    \text{s.t.}~~
    &(\ref{cons:F}),(\ref{cons:fij}),\notag
    \end{align}
    which is a convex optimization problem.
    By using the Lagrange multiplier method, a closed-form solution for the above problem can be obtained as 
    \begin{align}\label{eq:f_solution}
    f^\text{e}_{ij}(t)=\frac{\sqrt{Y_{ij}(t)x_{ij}(t)}}{\sum_{i\in\mathcal{I}}\sum_{j\in\mathcal{J}}\sqrt{Y_{ij}(t)x_{ij}(t)}}F.
    \end{align}
    Then, we obtain $K$ groups of decisions $\widetilde{y}_{j}(t)$, $\widetilde{z}_{j}(t)$, $x_{ij}(t)$, $f^\text{e}_{ij}(t)$.
    Then, we pick up the group that maximizes reward function $r(s_t,a_t)$ as the final action.
    % Next, according to the reward function $r(s_t,a_t)$, we can obtain the optimal decisions $y_{j}(t)$, $z_{j}(t)$, $x_{ij}(t)$, $f^\text{e}_{ij}(t)$ that achieve the best reward from the $K$ groups of decisions.
    Besides, with the rough idea given by the optimization stage and this optimization module, the DNN of actor network can reduce the number of output decisions from $2J+I+IJ$ decisions (for $\widetilde{y}_{j}(t)$, $\widetilde{z}_{j}(t)$, $x_{ij}(t)$, $f^\text{e}_{ij}(t)$) to $I$ decisions (only for $x_{ij}(t)$).
\end{itemize}
    
\subsubsection{Critic Network}
% With the policy of the actor network, the state $s_t$ will evolve to states $s_{t+1}...,s_{T-1}$ for time slots $t+1,...,T-1$, respectively. 
% The updating network will use the critic network to estimate the \textit{state values} of some selected states from $\{s_t,s_{t+1}...,s_{T-1}\}$. 
% Here for a state, the state value describes the expected subsequent cumulative reward of the system as evolving from the state. 
% Next, we use $s_t$ as an example to show how the critic network estimates its state value.

The critic network estimates the \textit{state value} of system states. Here for a state, the state value describes the expected subsequent cumulative reward of the system as evolving from the state.
% Then, given the estimated state value, the actor network could be improved to output a better policy with more rewards.
The critic network consists of a DNN with the same structure as the actor network except for the output layer. The DNN receives the input state, say $s_t$, and outputs the estimated state value, which is defied as $v_{\phi}:s_t\mapsto v_{\phi}(s_t)$, and $\phi$ is the parameters of DNN. 

\subsubsection{Updating Network}
The updating network is utilized to train the DNNs' parameters in the actor and critic networks.
The training process is shown as follows.
The agent interacts with the environment continuously using the policy given by the actor network and wins the reward from the environment, and then the environment evolves to the next state.
During the interaction of each time slot, the corresponding state, action, reward, and next state are stored as a \textit{transition} in the \textit{replay buffer}.
After every $N$ episodes, we sample a batch of transitions from the replay buffer to update the networks.
Given the transitions, the critic network estimates the state value and calculates the \textit{advantage}, and the actor network outputs the new policy.
Specifically, the advantage characterizes whether the action is better or worse than the default behavior of the policy.
The generalized advantage estimation (GAE) \cite{schulman2015high} is adopted to construct the advantage function, which facilitates the variance reduction and training stability.
The advantage at time slot $t$ is denoted as $\hat{R}(t)$ and given by
\begin{align}\label{eq:advantage}
\hat{R}(t)=\sum_{l=0}^{T-1-t} (\gamma\lambda)^l \hat{\delta}(t+l),
\end{align}
where $\gamma\in[0,1]$ is the discount factor in the discounted cumulative reward, $\lambda\in[0,1]$ is the parameter that enables bias-variance tradeoff, and $\hat{\delta}(\cdot)$ is the temporal difference (TD) error \cite{TDerror} and given by
%\begin{align}
$\hat{\delta}(t)=r(s_t,a_t)+\gamma v_{\hat{\phi}}(s_{t+1}) - v_{\hat{\phi}}(s_t),$ %\nonumber
%\end{align}
where $\hat{\phi}$ denotes the DNN's parameters of last training round.

Based on this, the loss function is constructed to update the actor and critic networks.
We have the actor loss, denoted as $l_t^\text{act}(\theta)$, and the critic loss, denoted as $l_t^\text{cri}(\phi)$, respectively.
% The loss function consists of actor loss $l_t^\text{act}(\theta)$ and critic loss $l_t^\text{cri}(\phi)$.

Specifically, the actor loss $l_t^\text{act}(\theta)$ contains the clip loss, denoted as $l_t^\text{clip}(\theta)$, and the entropy bonus, denoted as $l_t^\text{ent}(\theta)$.
Let $\pi_{\theta}(x_{ij}|s_t)$ denote the probability of outputting offloading decision $x_{ij}$ given state $s_t$ under policy $\pi_{\theta}$, and $\hat{\theta}$ denote the DNN's parameters of the actor network in the last training round.
We define the probability ratio as $\varphi_t(\theta)\triangleq\frac{\pi_{\theta}(x_{ij}|s_t)}{\pi_{\hat{\theta}}(x_{ij}|s_t)}$ that describes the ratio between the new policy and old policy, and thus $\varphi_t(\hat{\theta})=1$. Then, the clip loss $l_t^\text{clip}(\theta)$ is given by
\begin{align}\nonumber
l_t^\text{clip}(\theta)=\mathbb{E}[\min\{\varphi_t(\theta)\hat{R}(t),\text{clip}(\varphi_t(\theta),1-\epsilon_\pi,1+\epsilon_\pi)\hat{R}(t)\}],
\end{align}
where $\text{clip}(\varphi_t(\theta),1-\epsilon_\pi,1+\epsilon_\pi)$ is the clip function that truncates $\varphi_t(\theta)$ between the range $[1-\epsilon_\pi,1+\epsilon_\pi]$, and $\epsilon_\pi\in[0,1]$ is the hyperparameter that adjusts the clip range. 
The clip loss is designed to maximize the advantage while preventing the policy updates from drifting beyond a predetermined interval.
Next, the entropy bonus $l_t^\text{ent}(\theta)$ is defined as
\begin{align}\nonumber
l_t^\text{ent}(\theta)=-\sum_{x_{ij}\in\{0,1\},i\in\mathcal{I},j\in\mathcal{J}}
\mathbb{E}[\pi_{\theta}(x_{ij}|s_t)\ln\pi_{\theta}(x_{ij}|s_t)],
\end{align}
which is the entropy of the action distribution. The maximization of entropy bonus encourages policies to be more stochastic, promotes exploration, and prevents premature convergence to suboptimal policies.
Then, $l_t^\text{act}(\theta)$ is expressed as 
\begin{align}\label{eq:loss_actor}
l_t^\text{act}(\theta)=l_t^\text{clip}(\theta)+\beta l_t^\text{ent}(\theta),
\end{align}
where $\beta$ is the hyperparameter for the entropy bonus.
% In addition, we set the entropy hyperparameter $\beta$ to decay so that the actor gradually converge from exploration to the optimal decision.

In addition, we also utilize a clip function to construct the critic loss $l_t^\text{cri}(\phi)$.
We denote the clipped state value under state $s_t$ as $v^{\text{clip}}_{\phi}(s_t)$, which is expressed as 
\begin{align}\nonumber
v^{\text{clip}}_{\phi}(s_t)=v_{\hat{\phi}}(s_t)+\text{clip}(v_{\phi}(s_t)-v_{\hat{\phi}}(s_t),-\epsilon_v,\epsilon_v),
\end{align}
where $\text{clip}(v_{\phi}(s_t)-v_{\hat{\phi}}(s_t),-\epsilon_v,\epsilon_v)$ limits the gap between the new state value and the old state value to be within range $(-\epsilon_v,\epsilon_v)$, and $\epsilon_v>0$ is the hyperparameter determining the clip range.
In this way, the variability can be reduced during the training of the critic network.
Then, we denote the target state value under state $s_t$ as $v^{\text{tar}}_{\hat{\phi}}(s_t)$, which is the sum of previously estimated state value and the advantage, and is given by
%\begin{align}\nonumber
$v^{\text{tar}}_{\hat{\phi}}(s_t)=v_{\hat{\phi}}(s_t)+\hat{R}(t).$
%\end{align}
For reducing the estimation error, we construct the critic loss $l_t^\text{cri}(\phi)$ as the following squared error loss: 
\begin{align}\label{eq:loss_critic}
l_t^\text{cri}(\phi)=&\mathbb{E}[\max\{(v_{\phi}(s_t)-v^{\text{tar}}_{\hat{\phi}}(s_t))^2,\nonumber\\&(v^{\text{clip}}_{\phi}(s_t)-v^{\text{tar}}_{\hat{\phi}}(s_t))^2\}]. 
\end{align}

Next, we perform multi-step stochastic gradient ascent for the actor loss $l_t^\text{act}(\theta)$ and multi-step stochastic gradient descent for the critic loss $l_t^\text{cri}(\phi)$, respectively.
After the training, the actor network can be applied to output offloading decision $\boldsymbol{x}(t)$, downloading decision $\boldsymbol{y}(t)$, caching decision $\boldsymbol{z}(t)$, and computing rate allocation decision $\boldsymbol{f}(t)$.
% Besides, during the training process, the learning rate is set to decay, which helps to increase performance \cite{Engstrom2020Implementation}, \cite{andrychowicz2021what}.

% \section{Extend Discussion}
% (? cooperative MEC)

% objective: delay+AoI+computing price+caching price

% alg: online algorithm, competitive ratio, proposed$\approx$local optimal with cached$>$local optimal without cached$>$better virtual$>$virtual$>$opt.

\section{Experimental Results}\label{Sec:simulation}
In this section, we conduct simulations for performance evaluation of the proposed OIODRL method.
All simulations are conducted on a Pytorch 2.1.2 platform with an Intel Core i7-13700K $\SI{3.40}{\giga\hertz}$ CPU and $\SI{32.0}{\giga\byte}$ memory. 
\subsection{Parameter Setting}
In simulations, we consider that the length of each time slot is $\SI{15}{\min}$ and totally we simulate $1152$ time slots ($12$ days). 
There are totally $10$ service types and $5$ users.
The data size of service is uniformly distributed from $\SIrange{2}{6}{\giga\byte}$, the purchasing price of a service is uniformly distributed in $[1, 50]$ unit, and the refreshing price of a service is set to be one-tenth of the purchasing price of the service.
For each task from users, the data size follows the truncated Gaussian distribution from $\SI{500}{\mega\byte}$ to $\SI{2}{\giga\byte}$, and the data size of task processing result is one-tenth that of the task data size.
Each byte of task requires $330$ CPU cycles for processing.
By default, the storage capacity and computing capacity of the ES is $\SI{16}{\giga\byte}$ and $\SI{5.4e9}{\hertz}$, respectively.
The transmission rate between the ES and the CS is $\SI{25}{\mega\byte}/\rm{s}$.
The computing rate allocated by the CS to each task is $\SI{2e9}{\hertz}$.
% Besides, we set the AoI of service at the CS as follows:
% If the data update of service $j$ reaches the CS at time slot $t$, define the update indicator $\mu_j(t)=1$; otherwise, $\mu_j(t)=0$. We have
% \begin{equation}\nonumber
% A_{j}^\text{c}(t)=
% \begin{cases}
% A_{j}^\text{c}(t-1)+1,&\text{~if~}\mu_j(t)=0,\\
% \alpha_{j}^\text{c}(t),&\text{~if~}\mu_j(t)=1,
% \end{cases}
% \end{equation}
% where $\alpha_{j}^\text{c}(t)$ is the update transmission delay from sampling end to CS that possibly in milliseconds owing to the small amount of real-time information, such as temperature and humidity.
% Thus, we consider $\alpha_{j}^\text{c}(t)$ is much less than the length of one time slot.
% When the service $j$'s update arrives at CS, its AoI drops to $\alpha_{j}^\text{c}(t)$, else it increases by one time slot on top of the original one.
% Then, we set the possibility of update indicator $\mu_j(t)=1$ is sampled uniformly in $[25\%, 50\%]$ and the update transmission delay $\alpha_{j}^\text{c}(t)$ varies uniformly in $[0.05, 0.1]$ time slot. 

Besides, for each service, the CS receives an update at each time slot with $25\%$ probability, and the services are updated independently from each other. 
Moreover, the threshold value of the long-term average AoI ($A_j^{\max}$) is set to be uniformly distributed in $[5, 10]$ time slots by default.
Besides, we set the Lyapunov parameter $V$ to $1$ by default and set the cost weights $\lambda_D=0.1$, $\lambda_c=1$, $\lambda_p=1$, $\lambda_s=10$. 

In the proposed OIODRL method, we use the same structure of DNNs in both the actor and critic networks. The adopted DNN consists of $9$ full connected layers with $2048\times 2048$ neurons in the largest layer, and uses GELU activation function and dropout method with probability $50\%$ as a regularization technique and to avoid co-adaptation of neurons.
The parameters in the advantage calculation of (\ref{eq:advantage}) are set to $\gamma = 0.8$ and $\lambda = 0.95$, the clip rates $\epsilon_\pi$ in $l_t^\text{clip}(\theta)$ and $\epsilon_v$ in $v^{\text{clip}}_{\phi}(s_t)$ are both set to $0.2$.
During the training process, the batch size is $256$, and the initial learning rate is $1\times 10^{-3}$ and gradually decays to $1\times 10^{-4}$. Besides, the Adam optimizer is used when the actor and critic networks are updated.

\subsection{Verification of Advantage of PPO Algorithm}
To illustrate the superiority of the PPO algorithm in the proposed OIODRL, we implement the following DRL algorithms to replace the PPO algorithm in our OIODRL method: 
\begin{itemize}
    \item \textit{Advantage actor critic (A2C)} \cite{mnih2016asynchronous}: The structure of the A2C algorithm is the same as that of the PPO algorithm in our proposed OIODRL. To implement A2C, the loss function of the actor network in (\ref{eq:loss_actor}) is replaced by    
    \begin{align}\nonumber
    l_t^\text{act}(\theta)=\mathbb{E}[\log\pi_{\theta}(x_{ij}|s_t)\hat{R}(t)]+\beta l_t^\text{ent}(\theta),
    \end{align}
    which contains the advantage-based loss and the entropy bonus, % The multi-step stochastic gradient ascent is performed to maximize the objective $l_t^\text{act}(\theta)$.
    and the loss function of the critic network in (\ref{eq:loss_critic}) is replaced by 
    \begin{align}\nonumber
    l_t^\text{cri}(\phi)=\mathbb{E}[(v_{\phi}(s_t)-v^{\text{tar}}_{\hat{\phi}}(s_t))^2].
    \end{align}
    % which is minimized by the multi-step stochastic gradient descent.
    
    \item \textit{Deep Q-network (DQN)} \cite{mnih2013playing}: The DQN algorithm uses a neural network to estimate the action-value function, which is called a Q-network. The loss function of the Q-network at time slot $t$, denoted as $l_t^\text{DQN}$, is give by
    \begin{align}\nonumber
    l_t^\text{DQN}=\mathbb{E}[r(s_t,a_t)+\rho\max_{a_{t+1}} Q(s_{t+1}, a_{t+1})-Q(s_{t}, a_{t})],
    \end{align}
    where $\rho\in[0,1]$ is the discount factor of future rewards and designed to be the same as $\gamma$ in (\ref{eq:advantage}), and $Q(s, a)$ is the action-value function given state $s$ and action $a$.
    % The Q-network is updated by minimizing the loss function based on stochastic gradient descent.
    % For more exploration, the action having the maximum $Q$ value is chosen with probability $90\%$ and the action is randomly chosen with probability $10\%$.
\end{itemize}

Fig. \ref{fig:DRL} shows the learning curves of OIODRL with PPO, A2C, and DQN. 
PPO, A2C, and DQN converge at around $50$, $80$, and $20$ iterations, respectively.
Compared with A2C, PPO converges to a higher reward and is more stable in learning process. Compared with DQN, PPO achieves a much higher reward.  
%It can be seen that although the DQN converges the fastest, it falls into a locally optimal solution.
%The PPO achieves the best reward than the others. 
%Besides, PPO is more stable than A2C during the training process, since PPO adopts the clipped probability ratio which avoids a dramatic policy update.
% Fig. \ref{fig:DRL}~(b) depicts the learning curves of OIODRL with PPO with different setting. Specifically, `PPO' refers to the DRL algorithm used in the OIODRL; `PPO (LR=1e-3)' refers to the learning rate in the used PPO algorithm is set to $1\times10^{-3}$ and not decay; `PPO (LR=1e-4)' refers to the learning rate in the used PPO algorithm is set to $1\times10^{-4}$ and not decay; `PPO (no entropy)' refers to the loss function in the used PPO algorithm does not contain the entropy bonus.
% It can be seen that the others except PPO are unable to reach a satisfactory performance.
% The learning rate of $1\times10^{-3}$ is too large, then the algorithm falls into a local optimum; the learning rate of $1\times10^{-4}$ is too small, then it is too slow for convergence; without entropy bonus, the algorithm is not able to explore optimal decisions.

% \begin{figure}[t]
%     \centering
%      \subfloat[Different DRL algorithms.]{\includegraphics[width=0.5\linewidth]{fig/experiment/convergence.pdf}}
%      \subfloat[Different settings of PPO.]{\includegraphics[width=0.5\linewidth]{fig/experiment/PPO_LR_entropy.pdf}}
%     \caption{Learning curves with different DRL algorithms and PPO settings: reward versus the number of training iterations.}
%     \label{fig:DRL}
% \end{figure}

\begin{figure}[t]
    \centering
     \includegraphics[width=0.7\linewidth]{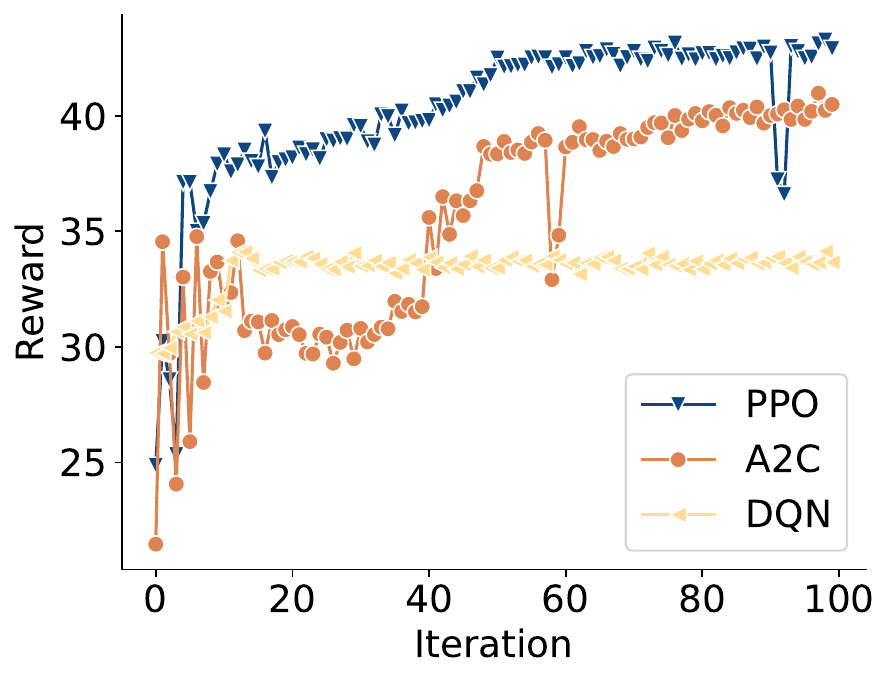}
    \caption{Learning curves with different DRL algorithms.}
    \label{fig:DRL}
\end{figure}

\subsection{Performance Evaluation}
To demonstrate the performance of our proposed OIODRL, we use the following benchmark methods for comparison. Note that the benchmark methods use the same approaches as ours to get the downloading decision and computing rate allocation decision, i.e., by Lemma \ref{lemma_cases} and (\ref{eq:f_solution}), respectively. 
Thus, the benchmark and OIODRL methods are different in making caching decision and offloading decision.
\begin{itemize}
    \item \textit{OPTIMAL}: An exhaustive search method is used to obtain the optimal caching decision and offloading decision for problem $\mathcal{P}2$ with time complexity $\mathcal{O}(2^{N+M})$.
    
    \item \textit{PPO-only}: This method only adopts the PPO algorithm without the optimization stage to solve problem $\mathcal{P}2$. The structure of PPO-only is similar to that of the PPO algorithm used in OIODRL, except that the PPO-only has two actor networks, in charge of caching decisions and offloading decisions, respectively.
    The loss functions of the actor network for offloading decisions and the critic network are given by (\ref{eq:loss_actor}) and (\ref{eq:loss_critic}), respectively.
    The loss function of the actor network for caching decisions is similar to (\ref{eq:loss_actor}), which is denoted as $l_t'(\theta)$ and given by
    \begin{align}
    l_t'(\theta)=&\mathbb{E}[\min\{\varphi'_t(\theta)\hat{R}(t),\text{clip}(\varphi'_t(\theta),1-\epsilon_\pi,1+\epsilon_\pi)\hat{R}(t)\}\nonumber\\
    &-\beta\sum_{z_{j}\in\{0,1\},j\in\mathcal{J}}
    \pi_{\theta}(z_{j}|s_t)\ln\pi_{\theta}(z_{j}|s_t)
    ],\nonumber
    \end{align}
    where $\varphi'_t(\theta)\triangleq\frac{\pi_{\theta}(z_{j}|s_t)}{\pi_{\hat{\theta}}(z_{j}|s_t)}$ is the probability ratio between the new policy and old policy of caching decisions.
    
    \item \textit{SDP-only}: This method only uses SDR method without the learning stage to solve problem $\mathcal{P}2$. The solution of problem $\mathcal{P}3$ via an SDP optimization toolbox is probably infeasible according to Lemma \ref{lemma_relax}. The infeasible caching decisions and offloading decisions are rounded to binary decisions, and then constraints (\ref{cons:S}) and (\ref{cons:xy}) are verified. If the caching decisions exceed the storage capacity, a cached service is randomly selected to be removed until constraint (\ref{cons:S}) is satisfied. If constraint (\ref{cons:xy}) is not satisfied, set the corresponding offloading decision to $0$.
    
    \item \textit{Joint service caching, computation offloading and resource allocation (JSCR) \cite{9385953}}: JSCR adopts the SDR method and rounding down to obtain a feasible solution, and then uses the alternating optimization method to further find a better solution.
    
    \item \textit{FIXED}: This method caches the fixed service types and processes the tasks of fixed service types at the ES, regardless of the system state. In the subsequent simulations, FIXED always caches services $1$ and $2$, and processes tasks $1$ and $2$ at the ES, by default.
\end{itemize}

% \begin{figure}[t]
%     \centering
%     \includegraphics[width=0.75\linewidth]{fig/experiment/Caching_offloading_decision.pdf}
%     \caption{Caching and offloading decisions at the first $5$ time slots with different methods. The colored squares indicate that the corresponding decisions are $1$.}
%     \label{fig:decisions}
% \end{figure}

\begin{figure}[t]
    \centering
    \includegraphics[width=0.7\linewidth]{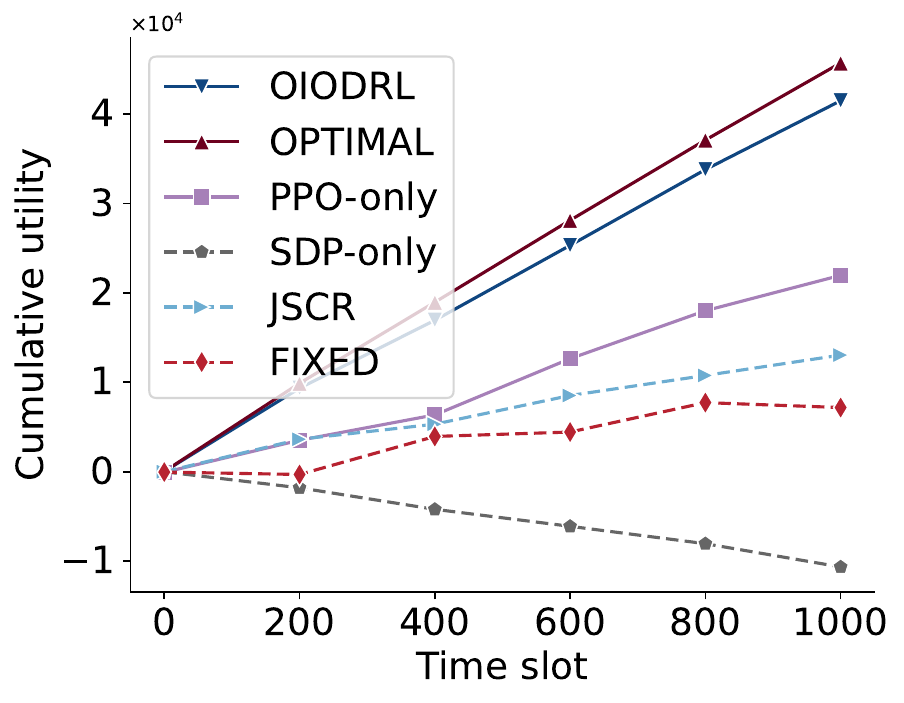}
    \caption{The performance of utility with different methods.}
    \label{fig:default_utility}
\end{figure}

% \begin{figure}[t]
%     \centering
%     \includegraphics[width=0.75\linewidth]{fig/experiment/default_aoi.pdf}
%     \caption{The performance with different methods in terms of service $1$'s AoI, i.e., $A_1(t)$.}
%     \label{fig:default_aoi}
% \end{figure}

% The caching and offloading decisions over the first $5$ time slots of the OIODRL and the other benchmark methods are shown in Fig. \ref{fig:decisions}
% The caching decisions of OIODRL are the same with OPTIMAL, except that OIODRL caches services $5$ and $6$ at time slot $4$.
% Besides, the caching decisions of SDP-only and JSCR are similar as JSCR also uses the SDR method and SDP optimization toolbox. 
% PPO-only caches the same services and does not change them over time, which suggests that PPO-only may be in a local optimum. 
% Similar results could be found in the offloading decisions.
% Also, OIODRL performs similarly to OPTIMAL.
% In addition, PPO-only and SDP-only processing of all tasks in the ES may overload the ES, and JSCR ameliorates this situation slightly.
% FIXED does not process any task because its fixed caching decisions do not hit any tasks.

The performance of cumulative utility with different methods is shown in Fig. \ref{fig:default_utility}. The proposed OIODRL performs similarly to the OPTIMAL and outperforms the others. 
It can be seen that PPO-only or SDP-only cannot achieve the best performance.
Besides, JSCR performs better than SDP-only, indicating that the alternating optimization method used in JSCR indeed improves the performance of SDR method while still not enough.
FIXED always makes fixed decisions and is not intelligent, and thus, obtains a low utility.

\begin{figure}[t]
    \centering
    \subfloat[Average reward.]{\includegraphics[width=0.5\linewidth]{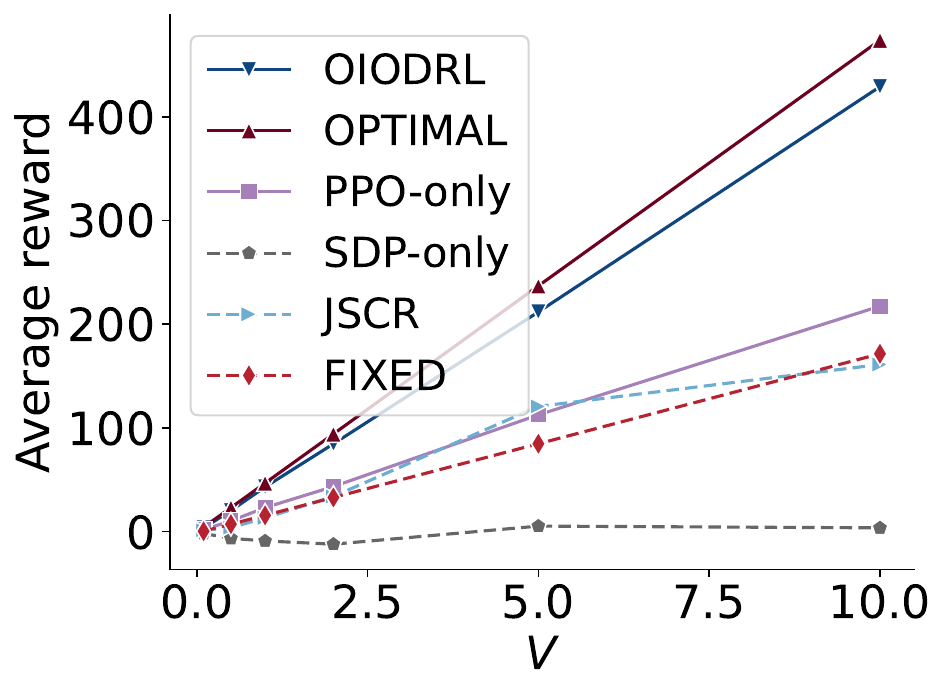}}
    \subfloat[Average utility.]{\includegraphics[width=0.5\linewidth]{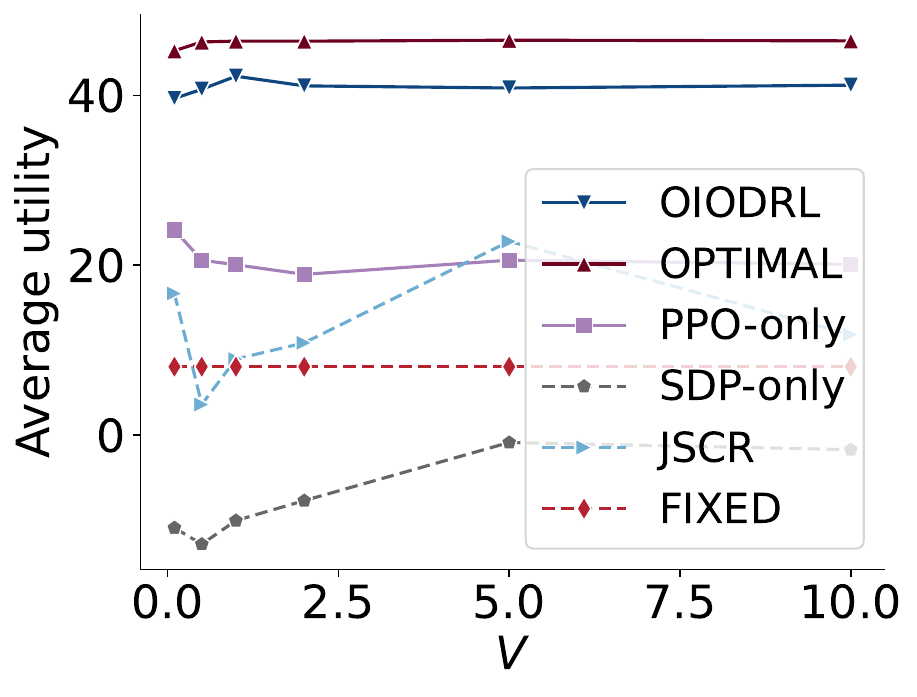}}
    \caption{The impact of $V$ on the reward and utility.}
    \label{fig:V_reward_utility}
\end{figure}

% \begin{figure}[t]
%     \centering
%     \includegraphics[width=0.75\linewidth]{fig/experiment/V_utility.pdf}
%     \caption{The impact of $V$ on the system utility with different methods.}
%     \label{fig:V_utility}
% \end{figure}

% \begin{figure}[t]
%     \centering
%      \subfloat[Different DRL algorithms.]{\includegraphics[width=0.5\linewidth]{fig/experiment/convergence.pdf}}
%      \subfloat[Different settings of PPO.]{\includegraphics[width=0.5\linewidth]{fig/experiment/PPO_LR_entropy.pdf}}
%     \caption{Learning curves with different DRL algorithms and PPO settings: reward versus the number of training iterations.}
%     \label{fig:DRL}
% \end{figure}

\begin{figure}[t]
    \centering
    \includegraphics[width=0.7\linewidth]{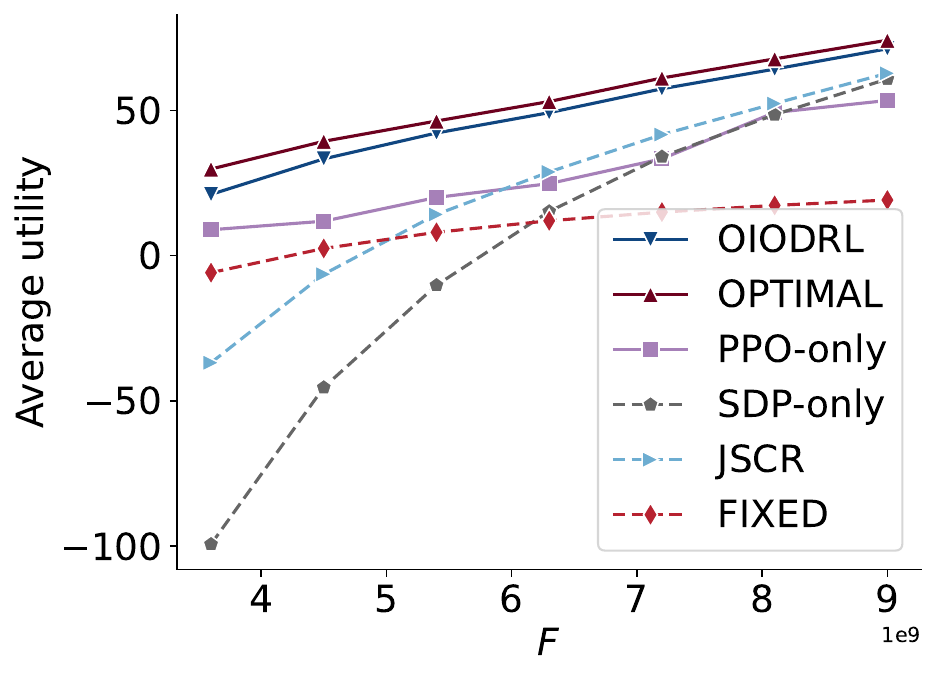}
    \caption{The impact of $F$ on the utility.}
    \label{fig:F_utility}
\end{figure}

\begin{figure}[t]
    \centering
    \includegraphics[width=0.7\linewidth]{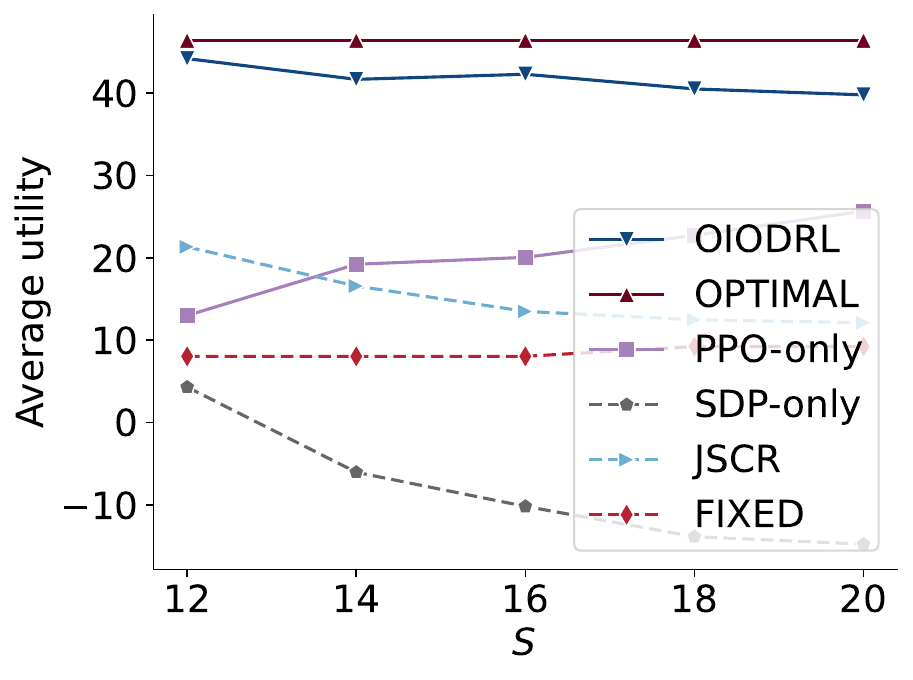}
    \caption{The impact of $S$ on the utility.}
    \label{fig:S_utility}
\end{figure}

\begin{figure}[t]
    \centering
    \includegraphics[width=0.7\linewidth]{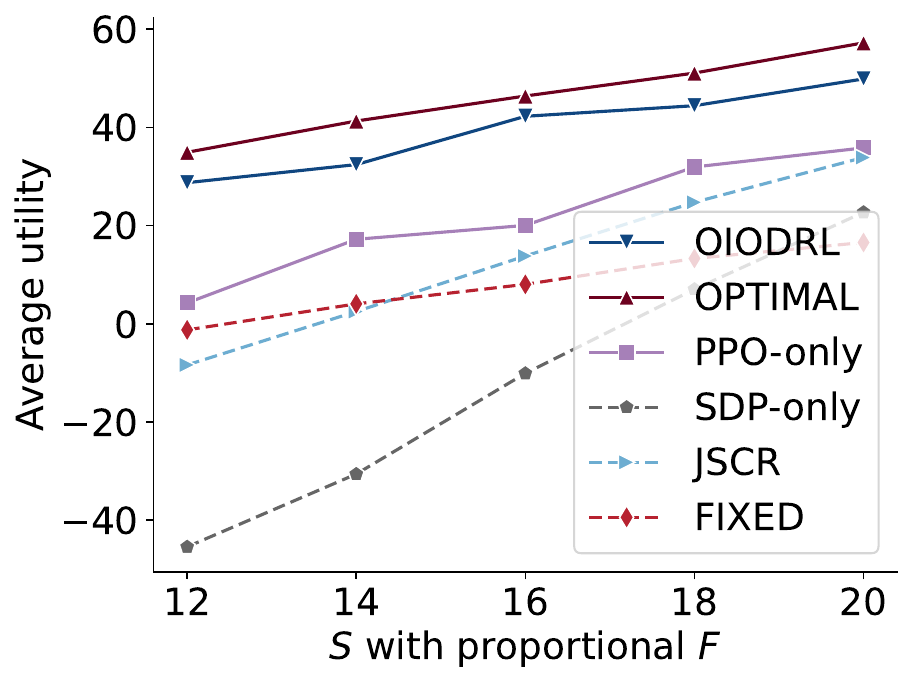}
    \caption{The impact of $S$ with proportional $F$ on the utility.}
    \label{fig:scaled_S_utility}
\end{figure}

Fig. \ref{fig:V_reward_utility}\,(a) shows the impact of $V$ (the weight parameter in Lyapunov drift-plus-penalty function) on the reward with different methods.
With the increasing of $V$, the rewards of all methods increase, since more rewards are assigned to the utility.
OIODRL performs similarly to OPTIMAL and outperforms other benchmark methods.
Fig. \ref{fig:V_reward_utility}\,(b) shows the impact of $V$ on the utility with different methods.
It can be seen that at the beginning of $V$ increase (more emphasis on the utility), both OPTIMAL and OIODRL show a brief increase on the utility. After $V=1$, they both converge. Therefore, we set the default value of $V$ to $1$.
FIXED always chooses the same strategy regardless of $V$, which results in the same low utility.
PPO-only, SDP-only, and JSCR are not stable when $V$ increases.

In Fig. \ref{fig:F_utility} shows the impact of $F$ (the computing capacity of the ES) on the utility with different methods.
As $F$ increases, the utilities of all methods increase, since the ES is more powerful to process the assigned tasks and the process delay is reduced.
For every $F$ value, OIODRL performs similarly to OPTIMAL and outperforms other benchmark methods. 

% Besides, the performance of both SDP-only and JSCR is significantly improved as $F$ increases.
% % As the analysis for Fig. \ref{fig:decisions}, 
% SDP-only and JSCR allocate an excessive number of tasks to the ES, resulting in overloading of the ES due to its limited computing resources. 
% Therefore, as $F$ increases, it is more powerful for the ES to process the assigned tasks, and then the utilities of SDP-only and JSCR are improved.

Fig. \ref{fig:S_utility} shows the impact of $S$ (the storage capacity of the ES) on the utility with different methods.
OIODRL performs close to OPTIMAL, and they achieve the best utilities. 
Note that the performance of OPTIMAL remains stable as $S$ increases. This stability is due to the fact that while $S$ increases without a corresponding increase in $F$, even if the ES has more storage capacity for caching services, it lacks sufficient computing resources to process the associated tasks.
Then, OPTIMAL will not cache more services as $S$ increases.
OIODRL performs similarly with a very slight decrease.
For PPO-only, a larger storage capacity may provide more exploration for a better utility. For SDP-only and JSCR, they cache more services when $S$ increases, and then process more tasks at the ES with insufficient computing resources, which leads to performance degradation.

Fig. \ref{fig:scaled_S_utility} shows the impact of $S$ with proportional $F$ (i.e., $S$ and $F$ vary with the same factor) on the utility with different methods.
With increased storage capacity and computing capacity of the ES, the performance of all methods improves. 
Furthermore, the performance of OIODRL consistently approaches that of OPTIMAL and surpasses those of the other methods.

\section{Conclusion}\label{Sec:conclusion}
In this paper, we provide online low-latency and fresh services to users in MEC networks. %at the minimum system cost. 
Different from existing works, we consider the freshness of cached edge services. 
This problem is formulated as a joint online long-term optimization problem, which is non-convex and NP-hard. 
To manage online provisioning without future system information, we propose a Lyapunov-based online framework that decouples the long-term optimization problem into a number of per-time-slot subproblems. %This framework provides a theoretical performance guarantee, demonstrating a cost-freshness tradeoff of $[\mathcal{O}(1/V)$, $\mathcal{O}(V)]$.
For each temporally decoupled subproblem, we propose the OIODRL method, which combines the advantages of optimization techniques and DRL algorithms.
Extensive simulations demonstrate the advantages and superiority of OIODRL compared to other DRL algorithms and benchmark methods.

{
% \small
\footnotesize
% \bibliographystyle{IEEEtran}
% \bibliography{mybib}

\begin{thebibliography}{10}
\bibitem{10275092}
L.~Qin, H.~Lu, Y.~Lu, C.~Zhang, and F.~Wu, ``Joint optimization of base station
  clustering and service caching in user-centric mec,'' \emph{IEEE Trans. Mob.
  Comput.}, vol.~23, no.~5, pp. 6455--6469, 2024.

\bibitem{9462377}
Q.~Xia, W.~Ren, Z.~Xu, X.~Wang, and W.~Liang, ``When edge caching meets a
  budget: Near optimal service delivery in multi-tiered edge clouds,''
  \emph{IEEE Trans. Serv. Comput.}, vol.~15, no.~6, pp. 3634--3648, 2022.

\bibitem{9599706}
Z.~Xu~\textit{et al.}, ``Energy-aware collaborative service caching in a
  5g-enabled mec with uncertain payoffs,'' \emph{IEEE Trans. Commun.}, vol.~70,
  no.~2, pp. 1058--1071, 2022.

\bibitem{9983987}
C.~Hou, C.~Zhou, Q.~Huang, and C.~Yan, ``Cache control of edge computing system
  for tradeoff between delays and cache storage costs,'' \emph{IEEE Trans.
  Autom. Sci. Eng.}, vol.~21, no.~1, pp. 827--843, 2024.

\bibitem{10293006}
H.~Yan~\textit{et al.}, ``Customer centric service caching for intelligent
  cyber–physical transportation systems with cloud–edge computing
  leveraging digital twins,'' \emph{IEEE Trans. Consum.}, vol.~70, no.~1, pp.
  1787--1797, 2024.

\bibitem{9832009}
X.~Xu, Z.~Liu, M.~Bilal, S.~Vimal, and H.~Song, ``Computation offloading and
  service caching for intelligent transportation systems with digital twin,''
  \emph{IEEE Trans. Intell. Transp. Syst.}, vol.~23, no.~11, pp.
  20\,757--20\,772, 2022.

\bibitem{9808348}
X.~Dai~\textit{et al.}, ``Task offloading for cloud-assisted fog computing with
  dynamic service caching in enterprise management systems,'' \emph{IEEE Trans.
  Industr. Inform.}, vol.~19, no.~1, pp. 662--672, 2023.

\bibitem{9362261}
J.~Yan, S.~Bi, L.~Duan, and Y.-J.~A. Zhang, ``Pricing-driven service caching
  and task offloading in mobile edge computing,'' \emph{IEEE Trans. Wireless
  Commun.}, vol.~20, no.~7, pp. 4495--4512, 2021.

\bibitem{9852277}
J.~Chen, H.~Xing, X.~Lin, A.~Nallanathan, and S.~Bi, ``Joint resource
  allocation and cache placement for location-aware multi-user mobile-edge
  computing,'' \emph{IEEE Internet Things J.}, vol.~9, no.~24, pp.
  25\,698--25\,714, 2022.

\bibitem{10103637}
W.~Chu, X.~Jia, Z.~Yu, J.~C. Lui, and Y.~Lin, ``Joint service caching, resource
  allocation and task offloading for mec-based networks: A multi-layer
  optimization approach,'' \emph{IEEE Trans. Mob. Comput.}, vol.~23, no.~4, pp.
  2958--2975, 2024.

\bibitem{9815021}
H.~Zhou, Z.~Zhang, Y.~Wu, M.~Dong, and V.~C.~M. Leung, ``Energy efficient joint
  computation offloading and service caching for mobile edge computing: A deep
  reinforcement learning approach,'' \emph{IEEE Trans. Green Commun. Netw.},
  vol.~7, no.~2, pp. 950--961, 2023.

\bibitem{9789202}
W.~Fan~\textit{et al.}, ``Joint task offloading and service caching for
  multi-access edge computing in wifi-cellular heterogeneous networks,''
  \emph{IEEE Trans. Wireless Commun.}, vol.~21, no.~11, pp. 9653--9667, 2022.

\bibitem{10007043}
Z.~Xue, C.~Liu, C.~Liao, G.~Han, and Z.~Sheng, ``Joint service caching and
  computation offloading scheme based on deep reinforcement learning in
  vehicular edge computing systems,'' \emph{IEEE Trans. Veh. Technol.},
  vol.~72, no.~5, pp. 6709--6722, 2023.

\bibitem{Lyapunov}
M.~Neely, \emph{Stochastic Network Optimization with Application to
  Communication and Queueing Systems}.\hskip 1em plus 0.5em minus 0.4em\relax
  Morgan \& Claypool, 2010, vol.~3.

\bibitem{huang2020closer}
S.~Huang and S.~Onta{\~n}{\'o}n, ``A closer look at invalid action masking in
  policy gradient algorithms,'' \emph{arXiv preprint arXiv:2006.14171}, 2020.

\bibitem{anschel2017averaged}
O.~Anschel, N.~Baram, and N.~Shimkin, ``Averaged-dqn: Variance reduction and
  stabilization for deep reinforcement learning,'' in \emph{International
  Conference on Machine Learning}, 2017, pp. 176--185.

\bibitem{tucker2017rebar}
G.~Tucker, A.~Mnih, C.~J. Maddison, J.~Lawson, and J.~Sohl-Dickstein, ``Rebar:
  Low-variance, unbiased gradient estimates for discrete latent variable
  models,'' \emph{in Advances in Neural Information Processing Systems},
  vol.~30, 2017.

\bibitem{schulman2017proximal}
J.~Schulman, F.~Wolski, P.~Dhariwal, A.~Radford, and O.~Klimov, ``Proximal
  policy optimization algorithms,'' \emph{arXiv preprint arXiv:1707.06347},
  2017.

\bibitem{10322784}
W.~Xu, J.~Yu, Y.~Wu, and D.~H.-K. Tsang, ``Joint channel estimation and
  reinforcement-learning-based resource allocation of
  intelligent-reflecting-surface-aided multicell mobile edge computing,''
  \emph{IEEE Internet Things J.}, vol.~11, no.~7, pp. 11\,862--11\,875, 2024.

\bibitem{schulman2015high}
J.~Schulman, P.~Moritz, S.~Levine, M.~Jordan, and P.~Abbeel, ``High-dimensional
  continuous control using generalized advantage estimation,'' \emph{arXiv
  preprint arXiv:1506.02438}, 2015.

\bibitem{TDerror}
R.~S. Sutton and A.~G. Barto, \emph{Reinforcement Learning: An
  Introduction}.\hskip 1em plus 0.5em minus 0.4em\relax Cambridge, MA: MIT
  Press, 1998.

\bibitem{mnih2016asynchronous}
V.~Mnih, A.~P. Badia, M.~Mirza, A.~Graves, T.~Lillicrap, T.~Harley, D.~Silver,
  and K.~Kavukcuoglu, ``Asynchronous methods for deep reinforcement learning,''
  in \emph{in International Conference on Machine Learning}, 2016, pp.
  1928--1937.

\bibitem{mnih2013playing}
V.~Mnih~\textit{et al.}, ``Playing atari with deep reinforcement learning,''
  \emph{arXiv preprint arXiv:1312.5602}, 2013.

\bibitem{9385953}
G.~Zhang~\textit{et al.}, ``Joint service caching, computation offloading and
  resource allocation in mobile edge computing systems,'' \emph{IEEE Trans.
  Wireless Commun.}, vol.~20, no.~8, pp. 5288--5300, 2021.

\end{thebibliography}

}

\end{document}